 \documentclass[12pt,draftcls,onecolumn]{IEEEtran}

\usepackage{kbordermatrix}
\usepackage{graphicx}
\usepackage{arydshln}
\usepackage{amssymb,amsmath}
\allowdisplaybreaks[1]
\usepackage{dsfont}
\usepackage{comment}
\usepackage[utf8]{inputenc}
\usepackage[T1]{fontenc}

\usepackage{cite}

\newtheorem{lemma}{Lemma}
\newtheorem{thm}{Theorem}
\newtheorem{cor}{Corollary}
\newtheorem{dfn}{Definition}

\providecommand{\abs}[1]{\left\lvert#1 \right\rvert}
\providecommand{\norm}[1]{\left\lVert#1 \right\rVert}
\providecommand{\argmin}[1]{\underset{#1}{\operatorname{argmin}}}

\providecommand{\vc}[1]{\boldsymbol{#1}}
\providecommand{\del}[1]{\widehat{#1}}

\renewcommand{\arraystretch}{1}

\newcommand{\diag}{\operatorname{diag}}

\begin{document}
\title{The Impact of Communication Delays on Distributed Consensus Algorithms}

\author{Konstantinos~I.~Tsianos and Michael~G.~Rabbat\thanks{The authors are with the Department of Electrical and Computer Engineering, McGill University, Montr\'eal, Canada. Email: konstantinos.tsianos@mail.mcgill.ca, michael.rabbat@mcgill.ca}\thanks{Portions of this work were previously presented at the \textit{49th Allerton Conf.~on Communication, Control, and Computing}~\cite{tsianosAllerton2011}.}}

\maketitle

\begin{abstract}
We study the effect of communication delays on distributed consensus algorithms. Two ways to model delays on a network are presented. The first model assumes that each link delivers messages with a fixed (constant) amount of delay, and the second model is more realistic, allowing for i.i.d.~time-varying bounded delays. In contrast to previous work studying the effects of delays on consensus algorithms, the models studied here allow for a node to receive multiple messages from the same neighbor in one iteration. The analysis of the fixed delay model shows that convergence to a consensus is guaranteed and the rate of convergence is reduced by no more than a factor $O(B^2)$ where $B$ is the maximum delay on any link. For the time-varying delay model we also give a convergence proof which, for row-stochastic consensus protocols, is not a trivial consequence of ergodic matrix products. In both delay models, the consensus value is no longer the average, even if the original protocol was an averaging protocol. For this reason, we propose the use of a different consensus algorithm called Push-Sum~[Kempe et al.~2003]. We model delays in the Push-Sum framework and show that convergence to the average consensus is guaranteed. This suggests that Push-Sum might be a better choice from a practical standpoint.
\end{abstract}

\section{Introduction}

This article aims to and understand the effects of communication delays on discrete-time distributed consensus algorithms. We build on two frameworks to model delay that were proposed in~\cite{tsianosAllerton2011}. For a simple model assuming fixed delays on the directed edges of a communication network, the question of how much the consensus convergence rate deteriorates in the presence of fixed delays was left open in~\cite{tsianosAllerton2011}. Here we prove that if the maximum delay on any edge is $B$, then the time to reach an $\epsilon$-accurate consensus in the delayed setting is no more than $O(B^2)$ iterations larger than that in the delay-free setting. For the fixed delay model, we generalize the construction of the random delay model presented in~\cite{tsianosAllerton2011} to use any arbitrary row stochastic consensus algorithm $P$ without delays. Our second major contribution is a formal convergence proof for the time-varying delay model. Finally, we show how both the fixed and random delay models can by used with a different consensus algorithm called Push-Sum consensus~\cite{PushSum}. For the random delay case we show that the delay model is simplified while convergence to the true average is still guaranteed. We conclude the paper with simulations that illustrate the effects of delays in distributed consensus computations. 

Our motivation to study communication delays comes from problems in distributed optimization and large-scale machine learning. The dramatic increase in available data has made the use of   parallel and distributed algorithms imperative for large problems (see for example \cite{langfordBook,boydAlternatingMultipliers}). Among numerous alternatives, a significant amount of research has focused on developing consensus based algorithms \cite{dualAveraging,distrStochSubgrOpt,boydAlternatingMultipliers,nedicDistributedOptimization,JohanssonIncrementalSubgrad} which combine some version of local optimization with a distributed consensus protocol running over a peer-to-peer  network. With such an approach, all computing nodes have the same role in the optimization procedure, thereby eliminating single points of failure and increasing robustness. This is important in large scale systems where machines may fail during the computation.  At the same time, consensus-based algorithms are simple to implement and avoid the bookkeeping required by algorithms using more structured routing. The consensus approach is also flexible and allows for adding more computational resources. On the other hand, peer-to-peer networks lack a highly organized infrastructure and coordinating the computing nodes becomes a challenge. Much of the recent analysis of consensus algorithms focuses on the case where communication is over a wireless network~\cite{gossipReview}.

For implementations of consensus-based optimization algorithms running on (wired) compute clusters, the issue of communication delays arises quite naturally. For example, in typical machine learning problems, the decision variable (and hence the message size) can quickly exceed many megabytes in size. During the time it takes to transmit such large messages, a modern processor can perform a significant amount of local processing of its own data, and the  received information always appears to be delayed. In addition, cluster computing resources are typically shared among many users, and delays to one task are introduced if processors devote some of their cycles to other unrelated tasks. Finally, any network infrastructure is bound to have some fluctuation in its performance for reasons beyond our control. It is thus important first to model communication delays, and then incorporate those models in the analysis of consensus algorithms to understand what the effects of delays will be. 

\subsection{Contributions}

In this article we study communication delays in discrete time and study their effects on convergence of consensus algorithms, focusing on distributed averaging. The main contributions of the paper are the following:

\textit{Consensus under Fixed Delays---The effect of delay on convergence rate:} Previous work \cite{tsianosAllerton2011} introduced a fixed delay model where transmissions over each directed link of a network experience some fixed amount of delay that does not exceed $B$. Starting with a doubly stochastic consensus protocol $P$ it was shown that consensus is still achieved in the presence of fixed delays at an exponential rate which depends on the second largest eigenvalue of $\del{P}$, the modified consensus algorithm accounting for delays.  In this paper we use geometric arguments to show that the rate of convergence does not get worse by more than an factor of $O(B^2)$. 

\textit{Random delay consensus under general row stochastic protocols:} Given a strongly connected graph $G$, in \cite{tsianosAllerton2011} a construction is given for building a matrix $\del{P}$ that describes the consensus updates on $G$ under the assumption that each message experiences a random amount of delay that does not exceed $B$ iterations.  Here, we generalize this model so that $\del{P} = \del{P}(P)$; i.e., $\del{P}$ is constructed from a given row stochastic consensus protocol $P$ defined on $G$ without delays. 

\textit{Random delay consensus---Convergence proof for row stochastic protocols:} If the initial protocol $P$ on a graph $G$ without delays is row stochastic, using the proposed random delay model, the consensus dynamics are captured by a sequence of matrices $\del{P}(t)$ which may contain all-zero rows. This means that although the consensus updates remain linear, convergence cannot be established based on standard theory for stochastic matrix products. Here we give a complete proof of convergence under this random delay model. 

\textit{Delays under Push-Sum consensus:} We study a different consensus algorithm called Push-Sum consensus\cite{PushSum} which uses column stochastic matrices. We show that convergence properties of Push-Sum are not affected in the presence of delays, and the aforementioned convergence results and bounds still apply. In particular, it is noteworthy that consensus on the average is guaranteed even in the presence of bounded random delays.

\subsection{Paper Organization}

The rest of the paper is organized as follows. We first summarize our notational conventions in Section~\ref{sec:notation}. Section~\ref{sec:previous} reviews related work and Section~\ref{sec:consensus} briefly reviews the consensus problem. The fixed delay model and related results are given in Section~\ref{sec:fixed}. Next, Section~\ref{sec:random} describes and analyzes the random delay model. Illustrative simulation results appear in Section~\ref{sec:simulations}, and the paper concludes in Section~\ref{sec:future} with a discussion of possible extensions and future work.


\subsection{Notation} 
\label{sec:notation}
We use bold to indicate vectors; e.g. $\vc{x}$. Time $t$ is always discrete and time dependence is shown as $\vc{x}(t)$. Vectors are indexed by subscripts, i.e., $x_i(t)$ or $[x(t)]_i$ when it is more clear. For a set of indices $S$, by $x_S$ we mean the entries of the vector $\vc{x}$ corresponding to the elements in $S$, and to index the range of indices from $i$ to $j$ in the vector $\vc{x}$ we use the notation $[\vc{x}(t)]_{i:j}$. Capital letters are used for matrices and we write $p_{ij}$ , $P(i,j)$ or $[P]_{ij}$ for the element in row $i$ and column $j$ of matrix $p$; we also write $[P]_{i,:}$ for the $i$-th row and $[P]_{:,j}$ for the $j$-th column. A matrix transpose is denoted by $P^T$. In many contexts we talk about a quantity such as a graph $G$ or a matrix $P$ and the corresponding quantity in the presence of delays. We write $\del{G}$ and $\del{P}$ to denote versions of $G$ and $P$ under the delay model. The vector of all ones is indicated by $\vc{1}$ and the vector of all zeros by $\vc{0}$. We use a subscript to show the dimension of the vector, as in $\vc{1}_n$, when it is not clear from the context. We also use the indicator function $\mathds{1}[event]$ which is equal to $1$ if the $event$ is true and zero otherwise. For a graph $G=(V,E)$ to talk about a directed edge from node $i$ to node $j$ we may use $(i,j)$ or $i \rightarrow j$ or just a superscript ${ij}$.

\section{Previous Work}
\label{sec:previous}


There is a rich literature on distributed averaging algorithms; see~\cite{gossipReview,saberReview} and references therein. A lot of effort has been focused on analyzing the rate of convergence to the average consensus \cite{consensusTsitsiklis}. The connection between consensus protocols and the convergence of Markov chains \cite{randomizedGossip} reveals that the spectral properties of the underlying network play an important role in the convergence rate. Of practical interest are asynchronous consensus algorithms. In \cite{BroadcastGossip} is it shown that using asynchronous broadcasts and forming convex combinations of incoming information guarantees convergence to the average only in expectation. For time-varying protocols, \cite{WeakErgodicityConditionsJadbabaie07} provides necessary conditions under which convergence is achieved while \cite{AsymptoticConsensusVal} characterizes the expectation and variance of the consensus value. Interestingly, in this paper we show that convergence to the true average under the same conditions for time varying protocols is guaranteed when using a different type of algorithm called Push-Sum\cite{PushSum,WeightedGossip}.

The main focus of this work is the effect of communication delays on consensus algorithms. For applications in partial differential equations, distributed control and multi-agent coordination  \cite{trecateConsensus,richardDelayOverview}   and \cite{saberDelays,delayedConsensus} analyze continuous-time delay models where all messages incur the same constant delay. Our motivation comes from applications in distributed optimization where both computation and communication happen in rounds and take a significant amount of time. For this reason we focus on discrete-time models. An early treatment of delays in discrete-time distributed averaging algorithms can be found in \cite{bertsekasParallel}, where it is proved that convergence is not guaranteed if delays are unbounded.  An analysis of conditions for convergence in the presence of delays is given in \cite{consensusTsitsiklis}. Closer to our work are \cite{delaysCao}, \cite{nedicConsensusDelays} and \cite{VaidyaDelays2} which model delays in discrete time for consensus problems by augmenting the state space with delay nodes. However, in  \cite{delaysCao} the value to which the consensus algorithm asymptotically converges is not characterized. The model in \cite{VaidyaDelays2} accumulates all the delayed information in a single delay node and does not allow for delivery of messages out of order. The model in \cite{nedicConsensusDelays} has the same expressive power as our random delay model, although the equation describing the consensus dynamics in~\cite{nedicConsensusDelays} does not allow for receiving multiple messages from the same sender in one iteration.

\section{Distributed Averaging}
\label{sec:consensus}

Assume each node $i \in V$ in a strongly connected network $G=(V,E)$ of $|V| = n$ nodes holds a value $v_i$. We stack the initial values in a vector $\vc{x}(0) = (v_1, \ldots, v_n)^T$. The general consensus problem asks for a distributed algorithm such that the nodes of the network exchange messages with their neighbours and update their state to reach consensus i.e., $\vc{x}(t) \rightarrow c \vc{1}$ as $t \rightarrow \infty$. In other words, we want the nodes to agree on a common value $c$ using only local communication. It follows from Perron-Frobenius theory\cite{SenetaBook} that if we choose a row stochastic matrix $P$ that respects the structure of  the graph in the sense that $p_{ij} \neq 0$ if $(j,i) \in E$, consensus is achieved by the iteration
\begin{align} \label{eq:consensus_intro}
\vc{x}(t) = P \vc{x}(t-1) = P^t \vc{x}(0).
\end{align}
The reason is that $P \vc{1} = \vc{1}$ and $\vc{1}$ is the unique eigenvector corresponding to the eigenvalue $1$ while all the other eigenvalues have magnitude less than one and their contribution vanishes if we consider the eigendecomposition of $P^t$ as $t \rightarrow \infty$. As a result, $P^t$ converges to a rank-$1$ matrix where each row is equal to the stationary distribution $\vc{\pi}$ of the Markov chain associated with $P$. In the special case where $\vc{1}^T P = \vc{1}^T$, the matrix $P$ is doubly stochastic and the consensus value is $c = \frac{1}{n}$; i.e., consensus is achieved on the average. Some situations may require using a protocol which corresponds to a row stochastic update matrix $P$, e.g., because $G$ does not admit a doubly stochastic matrix~\cite{CortesDoublyStochDigraphs}. In such situations, if the stationary distribution $\vc{\pi}$ of $P$ is known in advance then consensus on the average can still be achieved by rescaling the initial values by $(n \vc{\pi}_i )^{-1}$\cite{TsitsiklisConvergenceSpeed09}. Reaching consensus on the average is particularly important in distributed optimization since, if consensus is achieved on a value other than the average, an undesired bias is introduced\cite{TsianosACC2012}.

When the protocol $P$ is fixed, the update \eqref{eq:consensus_intro} represents a \textit{synchronous} algorithm where all nodes transmit information to their neighbours at the same time and each node receives exactly one message from each neighbour at each iteration. If we want to model scenarios where nodes communicate asynchronously or, as we will see below, if we want to model random communication delays where information may arrive in a different order than it was transmitted and we receive an unknown number of messages from each neighbour, we must consider time-varying protocols $P(t)$. The situation now becomes more involved as we may not be able to specify the stationary distribution to which the algorithm converges beyond its mean and variance \cite{AsymptoticConsensusVal}. Furthermore if we restrict to protocols where each node only transmits information without expecting a response---i.e., one-directional communication---using time-varying doubly stochastic protocols becomes impossible without extra coordination, while row stochastic protocols only converge to the average in expectation\cite{BroadcastGossip}. For these reasons, in the following we also consider a different type of consensus algorithm called Push-Sum consensus which does not have these limitations in the time-varying case.

\section{Fixed Communication Delays}
\label{sec:fixed}

We first analyze a model where the delay over each communication link does not vary with time. This is generally not true in practice but a fixed delay model can be appropriate in an average sense when the true delay does not fluctuate too much. An open question in \cite{tsianosAllerton2011} for this model, is how does the convergence rate of consensus with fixed delays depend on the maximum delay $B$. After reviewing the fixed delay model, we provide an answer below. 

Note that for the rest of this section, whenever we talk about a quantity $Q$, such as a graph or a matrix, we use a hat (i.e., $\del{Q}$) for the transformed version of $Q$ in the presence of delays.

\subsection{Fixed Delay Model}
\label{sec:fixeddelays}

Assume that in a given network $G$, for a directed link $(i,j)$, every message from $i$ to $j$ is delayed by $b_{ij}$ time units. We model this delay by replacing the link $(i,j)$ with a chain of $b_{ij}$ virtual \textit{delay nodes} in the network, acting as relays between $i$ and $j$. This leads to a network $\del{G}$ which contains the original compute nodes, $V$, as well as $b= \sum_{(i,j) \in E} b_{ij}$ delay nodes. Our goal is to study the corresponding consensus protocol running over $\del{G}$. We assume that a consensus protocol $P$ in the delay-free network $G$ is given so in the presence of delays, the compute nodes still transmit and combine incoming messages using the weights provided by $P$. In  \cite{tsianosAllerton2011}, we describe how to construct a stochastic matrix $\del{P}$ in the augmented space of $n+b$ nodes starting from a delay-free consensus protocol $P$. The matrix $\del{P}$ encodes  communication of information between delay and compute nodes and has a stationary distribution $\del{\vc{\pi}}$ which is not uniform and depends on both $P$ and the edge delays. We clarify that the augmentation of $G$ with delay nodes is done just for the purpose of modelling and the analysis; no physical delay nodes are actually added to the network. 

To illustrate the construction of $\del{P}$ from $P$ , consider a  graph $G$ with $3$ nodes.  Suppose that the delay-free consensus protocol is specified by the matrix
\begin{equation} \label{eq:examplePorig}
P = \begin{bmatrix}
  \frac{2}{3} & \frac{1}{3} & 0              \\[0.3em]
  \frac{1}{6} & \frac{1}{3} & \frac{1}{2}    \\[0.3em]
  \frac{1}{6}           & \frac{1}{3} & \frac{1}{2}
\end{bmatrix}.
\end{equation}
To model a fixed delay of $2$ whenever node $1$ transmits to node $2$, we augment $G$ with two delay nodes $d_1^{1\rightarrow 2}, d_2^{1\rightarrow 2}$ so that information from $1$ to $2$ must pass through them first.
In the augmented graph $\del{G}$, the consensus protocol is described by a row stochastic matrix $\del{P}$. Using the rows of $P$ we write $\del{P}$ as
\begin{equation} \label{eq:exampleQ}
\del{P} = \kbordermatrix{
~ & 1 & 2 & 3 & d_{1}^{1\rightarrow 2} & d_{2}^{1\rightarrow 2} \cr
1 &  \frac{2}{3} & \frac{1}{3} & 0         & 0 & 0   \cr
2 &  0 & \frac{1}{3} & \frac{1}{2} & 0 & \frac{1}{6}  \cr
3 &   \frac{1}{6}           & \frac{1}{3} & \frac{1}{2} & 0 & 0   \cr
d_{1}^{1\rightarrow 2} &       1           & 0                 & 0           & 0 & 0   \cr
d_{2}^{1\rightarrow 2} &       0           & 0                 & 0           & 1 & 0  \cr
}.
\end{equation}
Each receiving node forms a convex combination of the incoming messages so in $\del{P}$, node $2$ receives information from node $d_2^{1 \rightarrow 2}$ with weight $\frac{1}{6}$ because $p_{2,1} = \frac{1}{6}$.

Using $\del{P}$ we can analyze the effect of delays on convergence based on the update equations for row stochastic consensus
\begin{align}
\del{\vc{x}}(t) = \del{P} \del{\vc{x}}(t-1),
\end{align}
where $\del{\vc{x}}(t)$ is the augmented state vector of dimension $n+b$ containing values for the compute nodes and virtual delay nodes. If $P$ is doubly stochastic, our previous work~\cite{tsianosAllerton2011} provides an exact characterization of $\del{\vc{\pi}}$, the stationary distribution of $\del{P}$. Let us index the directed edges of $G$ (without delays) by $r=1,2,\ldots,m$. We use the notation $\big(i(r), j(r)\big)$ to specify that edge $r$ starts at node $i$ and is directed to node $j$. Moreover, let $b_r$ denote the amount of delay on edge $r$, and with a slight abuse of notation, let $\del{\pi}_r$ denote the value of the stationary distribution vector for all delay nodes in the chain replacing edge $r$. The stationary distribution of $\del{P}$ has the structure
\begin{equation}
\del{\vc{\pi}} = [\del{\pi}_{V} \mathbf{1}_{n}^{T}\ \ \del{\pi}_{1} \mathbf{1}_{b_{1}}^{T}\ \cdots\  \del{\pi}_{m} \mathbf{1}_{b_{m}}^{T} ]^{T},
\end{equation}
and the exact values are 
\begin{align} \label{eq:stationary_distribution_fixed}
\del{\pi}_{V} = & \frac{1}{n + \sum_r b_{r} p_{i(r)j(r)}} ,\ \ \ \del{\pi}_{r} =  \frac{p_{i(r)j(r)}}{n + \sum_r b_{r} p_{i(r)j(r)}}.
\end{align}

In the special case where $P$ is a \textit{max-weight} doubly stochastic matrix\footnote{For an undirected graph $G$ without self loops, with adjacency matrix $A$ and node degrees $\vc{v} = [deg_1 \dots , deg_n]$ the max-weight matrix is defined as $P = I - \frac{diag(\vc{v}) - A}{\max_i deg_i + 1}$ and is doubly stochastic.}, the entries of $\del{\vc{\pi}}$ only take one of two values, one for the compute nodes in the set $V$ and one for the delay nodes i.e., it does not matter how the delays are distributed over the links. Specifically,  denoting by $C$ the set of delay nodes we have
\begin{align}
\del{\pi}_{V} & = \frac{d_{max}+1}{b + n(d_{max}+1)}, \ \ \  \del{\pi}_{C}  = \frac{1}{b + n(d_{max}+1)}
\end{align}
where $d_{max}$ is the maximum degree of $G$ viewed as undirected ignoring self-loops.


Notice that even when $P$ is doubly stochastic (and thus admits average consensus), the row stochastic delayed protocol $\del{P}$ does not converge to the average in general, since its stationary distribution is not uniform. To converge to the average with $\del{P}$ we need to rescale the initial values as explained in Section \ref{sec:consensus}, using the stationary distribution of $\del{P}$. 

By construction, the delay nodes only relay information and have no self loops. Thus, the diagonal entries in $\del{P}$ corresponding to delay nodes are zero. This makes $\del{P}$ a non-reversible Markov chain that is not strongly aperiodic\footnote{A Markov chain is strongly aperiodic if all the diagonal entries of its transition matrix are at least $1/2$.}, and the majority of known convergence rate results for Markov chains do not apply. To get a bound on the convergence rate under fixed delays, we apply the result from \cite{fillBounds} with the lazy version $\del{P}_{lazy} = \frac{1}{2}(I + \del{P})$ of $\del{P}$. First, the \textit{additive reversibilization} of a Markov chain with transition matrix $P$ is defined by:
\begin{align}
U(P) = \frac{P + \tilde{P}}{2},
\end{align}
where $\tilde{P}$ is the \textit{time-reversed} chain. Next, since $\del{P}_{lazy}$ is non-reversible but strongly aperiodic and converges no more than two times slower than $\del{P}$, applying Fill's result~\cite{fillBounds} we have
\begin{align} \label{eq:aq}
\norm{[\del{P}^{t}]_{i,:} - \del{\vc{\pi}}}_{TV}^{2} \leq & \norm{[\del{P}_{lazy}^{t}]_{i,:} -\del{\vc{\pi}}_{lazy}}_{TV}^{2}  \notag \\ 
\leq &\frac{(\lambda_{2}(U(\del{P}_{lazy})))^{t}}{4 [\del{\pi}_{lazy}]_i}
\end{align}
with $\del{\vc{\pi}}_{lazy} = \del{\vc{\pi}}$. 

Our initial work~\cite{tsianosAllerton2011} left open the question of to what extent delays effect the convergence rate of average consensus protocols. One way to address this is to understand how much larger is $\lambda_2(U(\del{P}_{lazy}))$ in comparison to $\lambda_2(P)$. We provide an answer next.

\subsection{Effect of Delays on Second Eigenvalue}

The convergence rate of a consensus protocol $P$ to stationarity in terms of total variation distance can be bounded by $\lambda_{2}(P)$, the second largest eigenvalue of $P$. The second eigenvalue in turn can be bounded using a geometric argument based on the \emph{Poincar\'{e}\ inequality}\cite{StookDiaconis,fillBounds}. The intuition is to look for the bottleneck edge which limits the flow of information and consequently the convergence speed. Assume the stationary distribution of $P$ is $\vc{\pi}$. For each pair of nodes $\{x,y\}$ of $G$, we choose a (directed) path $\gamma_{xy}$ from $x$ to $y$. To identify bottlenecks we look at how many paths $\gamma_{xy}$ go through the same edge. A measure of bottlenecks in $G$, is given by the \textit{Poincar\'e constant},
\begin{align} \label{eq:poincareK}
K = \max_{e=(v,w)}  \left[ \frac{1}{\pi_v p_{vw} } \sum_{x,y\ s.t. \ e \in \gamma_{xy}}  \abs{\gamma_{xy}} \pi_x \pi_y \right],
\end{align}
where $\abs{\gamma_{xy}} $ is the length (in number of edges) of the path $\gamma_{xy}$. The constant $K$ quantifies the load on the most heavily used edge. Less formally, that involves identifying an edge through which many and long paths must pass for pairs of nodes to communicate over $G$. In addition, the paths are assigned an importance based on the stationary distribution value at the endpoints. Depending on the quality of the paths, we get a more accurate characterization of bottlenecks. Given a set of paths $\Gamma = \{\gamma_{xy}\}$, the Poincar\'e constant gives a bound on the second eigenvalue of $P$:
\begin{align}
\lambda_{2} \leq 1 - \frac{1}{K}.
\end{align}
Our goal is to use a given set of canonical paths $\Gamma$ for $G$ to construct a set of canonical paths in $\del{G}$, the augmentation of $G$ after adding fixed edge delays. This will reveal how the delays effect the convergence rate we have for $P$. To that end, we compute the Poincar\'e constant for $\del{G}$ as a function of the Poincar\'e constant of the original graph $G$. 

Since $\del{P}$ represents a non-reversible Markov Chain, we consider the lazy additive reversibilization $U(\del{P}_{lazy})$ which is strongly aperiodic, reversible, has the same stationary distribution as $\del{P}$, and whose convergence rate bounds that of $\del{P}$. With the exception of some added self loops on the delay nodes, the graph structure compatible with $U(\del{P}_{lazy})$ is the same as that of $\del{P}$. To compute the Poincar\'e constant $\del{K}$ for $\del{G}$ we start with some observations and consequences of augmenting $G$ with fixed delays. We assume that the maximum delay on any edge is $B$ and we use subscripts to index the nodes on a delay chain. 

1. We claim that if $e=(v,w)$ is the bottleneck edge in $G$ with no delays, all edges on the delay chain $v \rightarrow d_{1} \rightarrow \cdots \rightarrow d_{B'} \rightarrow w, B' \leq B$, that replaces $e$ in $\del{G}$ are bottlenecks in $\del{G}$. The reason is that if a flow needs to go through $e$ in $G$, it will have to go through all of the delay edges replacing $e$ in $\del{G}$. This is true because the degrees of the compute nodes do not change by adding fixed delays the way we described above, and the paths between the compute nodes are just elongated without offering new path alternatives. As a result, to compute the Poincar\'e constant of $U(\del{P}_{lazy})$ we do not need to maximize over all edges in $\del{G}$. Instead we only examine edges in the middle of delay chains. That is, if a delay chain connecting compute nodes $a$ and $b$ has length $B'$, we only consider the edge $\del{e} = (d^{ab}_{\lfloor \frac{B'}{2}\rfloor} , d^{ab}_{\lfloor\frac{B'}{2}\rfloor+1})$.

2. \label{item:pathcases}
We intend to use the given collection of canonical paths $\Gamma$ on $G$ to derive a bound on the Poincar\'e constant of $\del{G}$. The graph with delays has more nodes and thus more paths to be considered. However, we can associate a collection of paths of $\del{G}$ with the same path in $G$ using the compute nodes as identifiers for each path. The key point is to ensure that if a path $\gamma_{xy}$ goes through an edge $e$ of $G$, then in $\del{G}$ we have a set of paths $\{\del{\gamma}_{xy} \}$ identified by the same compute nodes $x \rightarrow y$. All those paths go through $\del{e}$, the edge in the middle of the delay chain that replaced $e$ in $\del{G}$. By forming this path association, the expression for $K$ will appear in the bound for $\del{K}$. Figure \ref{fig:pathassociation} illustrates the path association.

\begin{figure}[t]
\begin{center}
\includegraphics[width=3.5in]{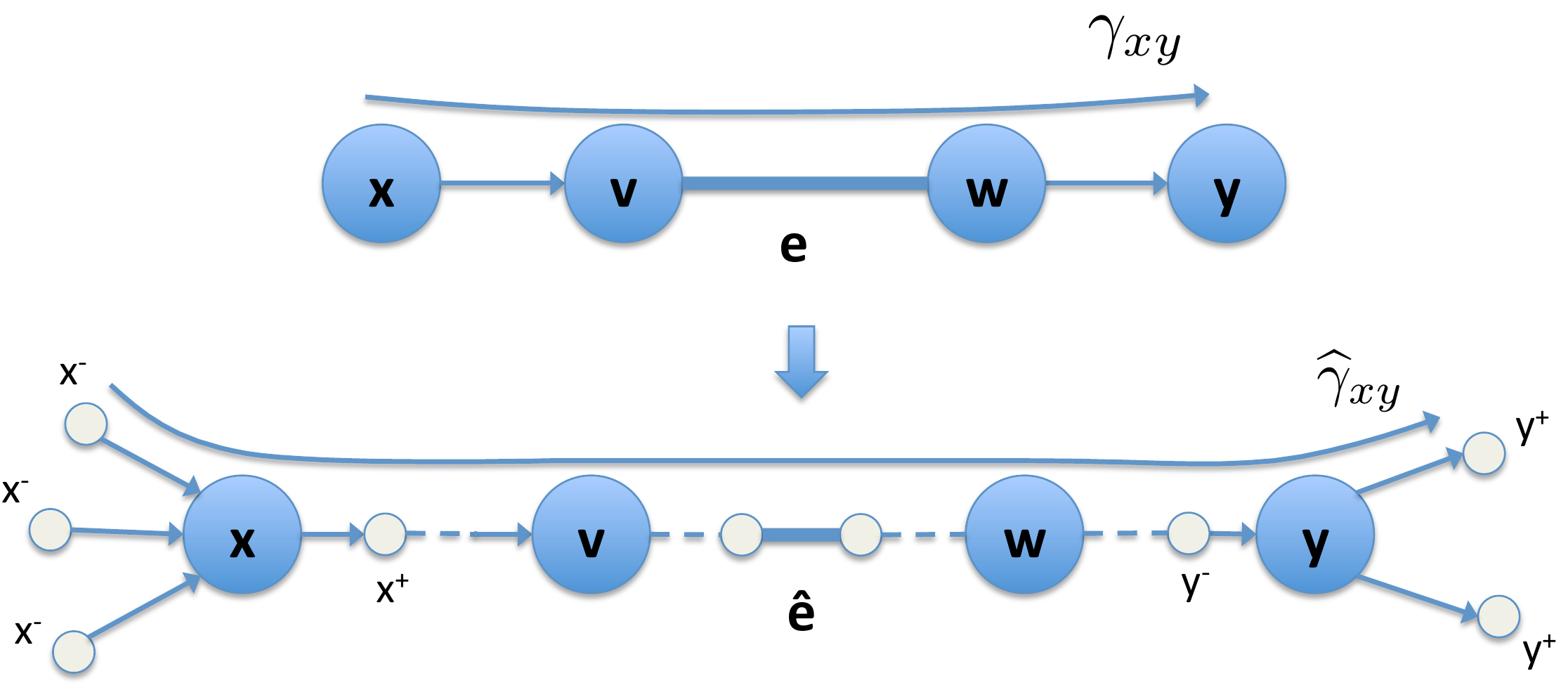}
\end{center}
\caption{\label{fig:pathassociation} (Top) A path $\gamma_{xy}$ in $G$. (Bottom) After adding delays in $\del{G}$, all paths from nodes $\{x^-, x, x^+\}$ towards nodes $\{y^-, y, y^+\}$ are associated with the same path $\gamma_{xy}$. If $e=(v,w)$ was a bottleneck edge in $G$, edge $\del{e}$ in the middle of the delay chain that replaced $e$ will be a bottleneck edge in $\del{G}$.}
\end{figure}

We distinguish the following nine cases. If $x,y$ are compute nodes in $\del{G}$, we associate $\del{\gamma}_{xy} \sim\gamma_{xy}$.  Note that $\abs{\del{\gamma}_{xy}} \leq (B+1) \abs{ \gamma_{xy}}$ when the maximum possible delay per edge is $B$. Next, to consider paths to or from delay nodes, we associate a delay node with the compute node that is closest to it in the direction of the path. Let us use the notation $x^{-}$ to denote delay nodes before $x$ associated with paths through $x$, and $x^{+}$ to denote delay nodes after $x$. For each path $\gamma_{xy}$ of $G$ going through edge $e$, we identify different cases of paths in $\del{G}$ going through $\del{e}$ (the middle edge in the delay chain that replaces $e$). We have eight possibilities:   $x \rightarrow  y^{-}, x \rightarrow y^{+},  x^{-} \rightarrow y^{-}, x^{-} \rightarrow y, x^{-} \rightarrow y^{+},  x^{+} \rightarrow y^{-}, x^{+} \rightarrow y,$ and $x^{+} \rightarrow y^{+}$. 

3. To get a cleaner expression for the bound, assume that $P$ is doubly stochastic. In that case, from \eqref{eq:stationary_distribution_fixed} we see that the stationary distribution of the compute nodes in the presence of delays is $\del{\pi}_x = \frac{\pi_x}{c}$ where $c = \frac{n + \sum_{r} b_{r} p_{r(i)r(j)}}{n}$. Moreover, for all compute nodes $x$, we have $\del{\pi}_x \geq p  \del{\pi}_{x^{-}}$ and  $\del{\pi}_x \geq p \del{\pi}_{x^{+}} $ where $p = \max_{i \neq j} p_{ij}$.

With the above considerations in mind, we start from the definition of the Poincar\'e constant for $\del{G}$:
\begin{align}
\del{K} = \max_{h=(a,b)}  \Big[ \frac{1}{\del{\pi}_a U(a,b) } \sum_{x,y\ s.t. \ h \in \del{\gamma}_{xy}}  \abs{\del{\gamma}_{xy}} \del{\pi}_x \del{\pi}_y \Big].
\end{align}
Let $e=(v,w)$ be a bottleneck edge of $G$. This means that the edge $\del{e}$ in the middle of the delay chain that replaces $e$ will be the bottleneck in $\del{G}$. After some algebra we can bound $\del{K}$ with an expression that involves $K$ (from \eqref{eq:poincareK}). Besides the leading constant involving the bottleneck edge, we need to break the sum over the canonical paths into summands according the nine cases we described in consideration $2$ above. We refer the reader to the appendix for a proof and we state here the final result. 

\begin{thm} \label{thm:poincare}
Let $G$ be a network endowed with a doubly stochastic consensus protocol $P$ and a set of canonical paths $\Gamma$ yielding a Poincar\'e constant $K$. Then adding fixed delays up to $B$ on the edges of $G$ yields a Poincar\'e constant $\del{K}$ for the delay graph $\del{G}$ for which
\begin{align}
\del{K} \leq Z K, \ \ \ Z  = &  \frac{ p_{vw}}{4 c} \Big[p^{2}(2d_{max}^{2} + 3d_{max}+1) B^{3} \notag\\  
&+ p(2pd_{max}^{2} + 2pd_{max} + 8d_{max} + 6) B^{2}  \notag \\
& + (8pd_{max} + p + 8) B + 8\Big],
\end{align}
where $(v,w)$ is a bottleneck edge in $G$, $p = \max_{i \neq j} p_{ij}$, $c = \frac{n + \sum_{r} b_{r} p_{r(i)r(j)}}{n}$ and $d_{max}$ is the maximum degree in the undirected graph $G$ ignoring self-loops.
\end{thm}

Theorem \ref{thm:poincare} yields a bound in the second eigenvalue and thus the spectral gap of $\del{P}$.

\begin{cor} \label{cor:poincare}
Suppose a doubly stochastic protocol $P$ on a graph $G$ has a spectral gap $1 - \lambda_2(P) \geq \frac{1}{K}$,  and assume that messages over the edges of $G$ experience arbitrary fixed delays of up to $B$ iterations. Then the spectral gap of $\del{P}$ is reduced by at most a factor $\Theta(B^2)$; i.e.,
\begin{align}
1 - \lambda_2(\del{P}) \geq \frac{1}{Z K}, \ \  Z = \Theta(B^2).
\end{align}
\end{cor}
\begin{proof}
From Theorem \ref{thm:poincare} we have  $\lambda_2(\del{P}) \leq \lambda_2(U) \leq 1 - \frac{1}{Z K}$. Since $b_r \leq B, r=1,2,\ldots, m$ we see that $c = \frac{n + \sum_{r} b_{r} p_{r(i)r(j)}}{n} =  \Theta(B)$ and thus $Z = \Theta(B^2)$. 
\end{proof}

To the best of our knowledge this is the first result to describe the effect of a bounded fixed delay on the convergence rate of average consensus. It shows that the delays cannot slow down consensus by more than a polynomial factor and convergence remains exponentially fast.

\section{Time Varying Communication Delays}
\label{sec:random}

To capture real network volatility, it is more appropriate to assume that link delays vary randomly with time. In \cite{tsianosAllerton2011}, a discrete-time random delay model is presented. However the construction only applies to uniform consensus weights (i.e., where $P$ is the natural random walk on $G$), and convergence to consensus is only verified in simulation. Here, we generalize the construction of the model from~\cite{tsianosAllerton2011} to use any row-stochastic protocol and we present a formal convergence proof.

\subsection{Random Delay Model}

Similar to the fixed delay model, we add virtual delay nodes. We assume again that delays are finite and upper bounded by a maximum delay $B$.  As emphasized in \cite{tsianosAllerton2011}, with random delays in discrete time we need to be careful. Others have previously analyzed a consensus update of the form
\begin{equation}
x_{i}(t+1) = \sum_{j=1}^{n} p_{ij} x_{j}(t - b_{ij}(t)),
\end{equation}
where $b_{ij}(t)$ is the random delay experienced by link $(i,j)$ at time $t$ \cite{saberDelays,nedicConsensusDelays}. However, this type of update implies that at time $t$ each node $i$ will only receive a single (possibly delayed) message from each neighbour $j$. In practice this may not be true. For example, take an edge $(i,j)$ whose delay could be $1$ or $2$. Assume at iteration $t$ node $i$ sends a message $m_{t}$ to $j$ and at time $t+1$, $i$ sends a new message $m_{t+1}$ to $j$. If $m_{t}$ is delayed by $2$ time units and $m_{t+1}$ is delayed by $1$ unit, then both $m_{t}$ and $m_{t+1}$ will be delivered to node $j$ at time $t+2$. This scenario can easily occur in practice when messages are large in size and receiving a message takes  a non-trivial amount of time during which a second message can arrive. When this happens, the receiving node polling its buffer  experiences the arrival of two messages during the same time slot. 

To model random bounded delays, we replace each directed edge of the original graph with multiple delay chains of varying lengths to model varying amounts of delay.  Every time a message is sent, a random decision is made for which delay chain the message will take to reach its destination\footnote{Of course in reality this random choice is made by the environment, i.e., the network, and is beyond our control. For modeling purposes to emulate and understand the effect of delays, we can draw a random sample from a distribution that we believe resembles how real network conditions fluctuate.}. If a communication network with $n$ computing nodes has  $m$ directed edges (not counting the self loops), each edge delivers messages with some bounded delay that is randomly chosen  between $0$ and $B$. For example for an edge $(i,j)$ with a maximum delay of $3$ we augment $(i,j)$ in $G$ with three parallel delay chains $(d_{1}^{1}), (d_{1}^{2}, d_{2}^{2}), (d_{1}^{3}, d_{2}^{3}, d_{3}^{3})$ in $\del{G}$; see Figure \ref{fig:add_random_delay}. We avoid indexing the delay nodes by edge number to not clutter notation. We augment the graph with $\frac{B(B+1)}{2}$ delay nodes per edge or $b = \frac{m B (B+1)}{2}$ delay nodes total, where $m$ is the number of edges in $G$. We also allow for messages to be delivered without delay, by including the directed edges $(i,j)$ of the original graph $G$. 
\begin{figure}
\begin{center}
\includegraphics[width=2.8in]{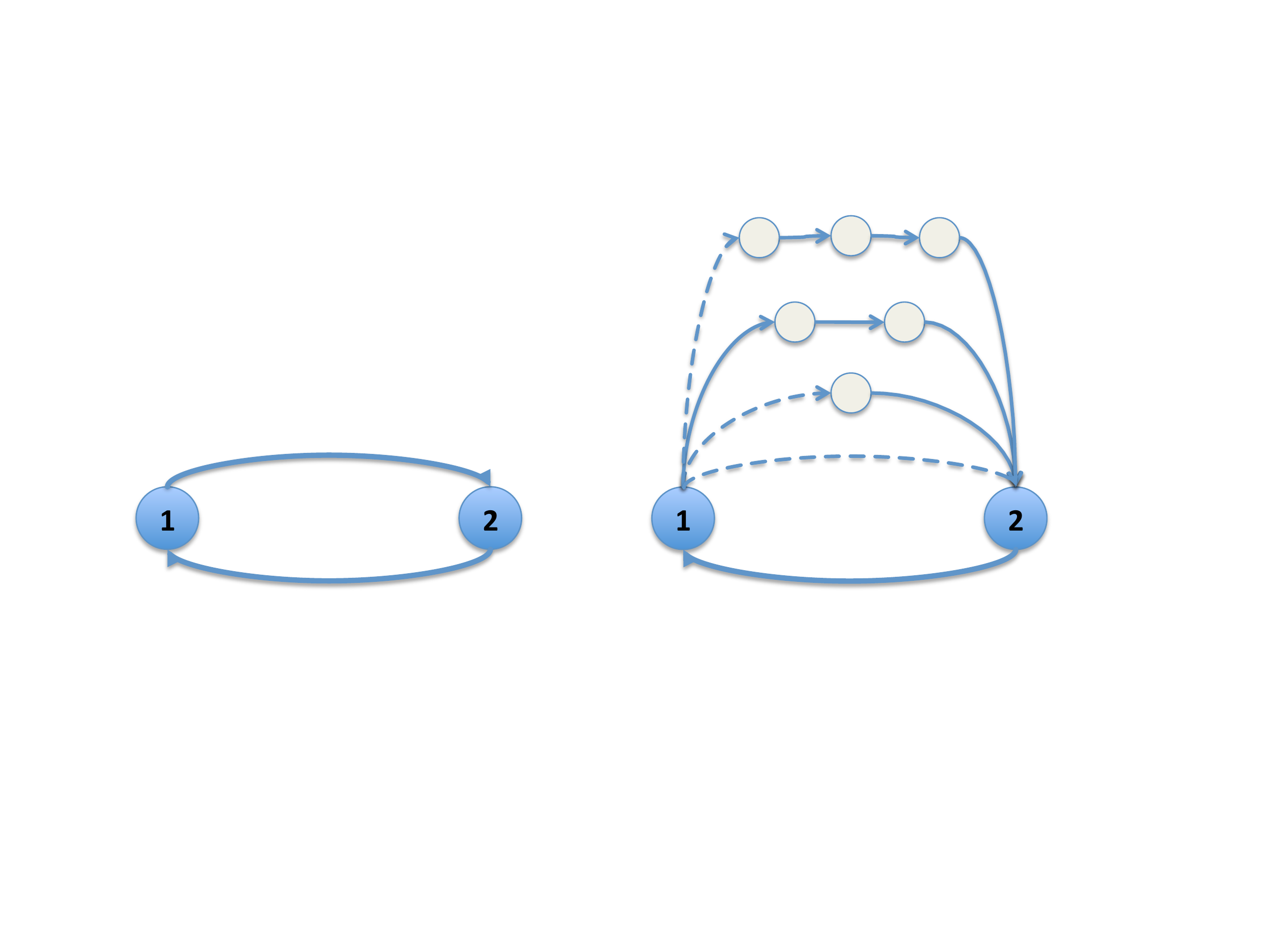}
\end{center}
\caption{\label{fig:add_random_delay} Adding a random bounded delay on edge $(1,2)$. At this particular instant, $1$ sends with delay  $2$ since the connections to delays $1$ and $3$ are deactivated.}
\end{figure}

Our goal is to write a matrix $\del{P}(t)$ that will describe the consensus dynamics under random delays using linear updates. Our previous work~\cite{tsianosAllerton2011} presented a model for the simple case where all incoming messages receive equal weight (proportional to the number of neighbors). To address the general case, we assume here that we are given a row stochastic protocol $P$ for the graph $G$, and we construct $\del{P}(t)$ using the weights suggested by $P$.

Every time a message is sent, it is routed randomly through one of the $B$ delay chains or the direct edge with zero delay. Outgoing edges to the other chains leading to the same recipient are cut off. Here we consider a time-varying delay model where each message experiences a delay that is i.i.d. from delays on other messages on different edges and different time moments. For more accurate modelling, we can impose any discrete probability distribution on the integers $0,\ldots,B$ to control the expected delay of an edge. This does not effect the convergence analysis presented below.

As we see, the augmented graph topology changes at every iteration based on which outgoing edges to delay chains are active. To describe the consensus update equations we need to model the changing topology. At each iteration, a delay is sampled for each message to be transmitted. Based on these delays, at iteration $t$ the graph adjacency matrix $A(t)$ is a sample from the set $\{A^{1}, \ldots, A^{(B+1)^{m}}\}$ of possible adjacency matrices.  Notice that a delay node could either contain a message or be empty, and a zero message is not the same as the node being empty. To keep track of which delay nodes are empty we define an indicator vector sequence $\{\phi(t)\}_{t=1}^{\infty}, \phi(t) \in \{0,1\}^{b}$. Using $A(t)$ and $\phi(t)$ we show how to write a transition matrix $\del{P}(t)$ at each iteration $t$.

We begin by noticing that adjacency matrices $A(t)$  have the structure
\begin{equation}
A(t) = \begin{bmatrix} \label{eq:randomdelayA}
	  I_{n \times n} + L(t)  &             J_{n\times b} \\
	   R(t)                  & C_{b \times b}
\end{bmatrix}.
\end{equation}
Matrix $A(t)$ should be interpreted as a directed graph adjacency matrix. Element $[A(t)]_{ij}$ is $1$ if there is a directed link from $j$ to $i$. Its constituent parts $L(t)$, $J_{n\times b}$, $R(t)$, and $C_{b \times b}$ are described next.

The upper left block is an identity matrix to represent the self-loops  plus a random $n \times n$ square matrix $L(t)$ with zeros on the diagonal and a one at position $(i,j)$ if compute node $j$ sends a message to compute node $i$ with zero delay\footnote{Note that zero delay means that a message sent at iteration $t$ will be delivered at iteration $t+1$, i.e., without any delay.} at iteration $t$. Matrix $R(t)$ is $b \times n$ and is also a random matrix. Whenever a compute node $i$ transmits to another compute node $j$ using delay chain $r =1,\ldots,B$, matrix $R(t)$ will encode that random delay choice for time $t$. For example, if at time $t$ node $j$ sends a message to $i$ which is delayed by $2$ steps (so that it will arrive at time $t+3$), $R(t)$ will contain a block for edge $(j,i)$ indicating the delay chain that is active, as illustrated in equation~\eqref{eqn:Rexample}.
\begin{equation}
R(t) = \kbordermatrix{
	~ 	&	1  	&	\cdots  	&	j  		& 	\cdots 	  & n  \cr
 \vdots 	& \vdots    &	        		&	\vdots 	&   			  & \vdots \cr
\cdashline{2-6}	 	
 d_{1}^{1} 	&	0	& \cdots	 	&	0		&	\cdots	  &	0 \cr
  d_{1}^{2} &	0	& \cdots	 	&	1		&	\cdots	&	0 \cr
 d_{2}^{2} 	&	0	& \cdots	 	&	0		&	\cdots	&	0 \cr
 d_{1}^{3} 	&	0	& \cdots		&	0		&	\cdots	 &	0 \cr
  d_{2}^{3} &	0	& \cdots		&	0		&	\cdots	&	0 \cr
 d_{3}^{3} 	&	0	& \cdots		&	0		&	\cdots	&	0 \cr
\cdashline{2-6}	
\vdots	&	 \vdots 	&	 	&	\vdots  &	  &	\vdots \cr}. \label{eqn:Rexample}
\end{equation}
Element $(d_1^2, j)$ of $R(t)$ is $1$ since $j$ will transmit to the first delay node in the chain of length $2$ towards $i$. The entries that are not shown within each block are all zero. 

Matrix $J_{n \times b}$ describes the connections between  the delay nodes $d_{r}^{r}$ at the end of each delay chain delivering messages to the compute nodes.  The part of $J_{n\times b}$ corresponding to the edge $(j,i)$ of $R(t)$ just discussed will look like
\begin{equation}
J_{n \times b}^T = \kbordermatrix{
	~                & 1 		& \cdots & j 	& \cdots & n \cr
	\vdots        &\vdots  	 & 	    & \vdots        & 	    &  \vdots \cr
	\cdashline{2-6}	
	 d_{1}^{1} & 0 & \cdots & 1       & \cdots  & 0 \cr
	 d_{1}^{2} & 0 & \cdots & 0       & \cdots  & 0 \cr
	 d_{2}^{2} & 0 & \cdots & 1       & \cdots  & 0 \cr
	 d_{1}^{3} & 0 & \cdots & 0       & \cdots  & 0 \cr
	 d_{2}^{3} & 0 & \cdots & 0       & \cdots  & 0 \cr
	 d_{3}^{3} & 0 & \cdots & 1       & \cdots  & 0 \cr
	 \cdashline{2-6}	
	\vdots        &\vdots  	 & 	    & \vdots        & 	    &   \vdots \cr
}.
\end{equation}
I.e., for edge $j \rightarrow i$, the entries $(j, d_1^1), (j, d_2^2)$ and $(j, d_3^3)$ in $A(t)$ are all $1$. Finally, we define the matrix $C_{b \times b}$ for forwarding messages from one delay node to the next on each chain. On a specific delay chain of length $h$, messages are forwarded through the action of an $h \times h$ Toeplitz forward shift matrix with $1$s on the first lower diagonal, i.e.,
\begin{equation}
S_{h} = \begin{bmatrix}
	0 		& 0 	& \cdots 	& 0	&  0	& 0 \\
	1 		& 0	&  		&  	&   	& 0 \\
	0 		& 1	&  		&     	&	& 0 \\
	\vdots 	& 	& \ddots 	& 	&	& \vdots \\
	0 		& 	&  		& 1	& 0   &  0 \\
	0 		& 0 	& \cdots 	&  0	& 1  	& 0 \\
\end{bmatrix}.
\end{equation}
For any edge $r=1,\ldots,m$, to forward messages through all delay chains we use a block diagonal matrix $K_{r} = \operatorname{diag} (S_{1}, S_{2}, \ldots, S_{B})$. Finally, since we have $m$ edges
\begin{equation}
C_{b \times b} = \operatorname{diag}(K_{1}, K_{2}, \ldots, K_{m}).
\end{equation}
Looking back at \eqref{eq:randomdelayA}, observe that every row of 
$[R(k) \ \ C_{b \times b}]$ contains at most one non-zero element and there are rows that are all zero.

Next, we define an indicator vector $\phi(t) \in \{0,1\}^{b}$  that keeps track of whether a delay node on any delay chain contains a message or is empty. Initially we have $\phi(0) = \mathbf{0}_{b}$. At iteration $t$, the first nodes in the delay chains may receive new information depending on which edges are activated by $R(t)$. The rest of the delay nodes will be non-empty only if their predecessors in the chains were non empty in the previous iteration. In other words, $\phi(t)$ evolves as
\begin{equation} \label{eq:phi}
\phi(t) = R(t) \mathbf{1}_{n} + C_{b \times b} \phi(t-1).
\end{equation}

After understanding the structure of the time-varying adjacency matrices $A(t)$, to describe the consensus transition matrices $\del{P}(t)$ we need to specify the weights used to combine incoming messages. Recall that each computing node might receive multiple messages from a neighbouring computing node, each arriving via a different delay chain. We will assign equal weights to all incoming messages from the same sender, and messages from different senders will receive weights according to $P$. For example, suppose compute node $i$ receives $\del{w}_{ij} + L_{ij}(t)$ messages from node $j$ where $0 \leq \del{w}_{ij} \leq B$ are the delayed messages and $L_{ij}(t) = 0$ or $1$ is a message without delay. Node $i$ will assign a weight $\frac{p_{ij}}{\del{w}_{ij} + L_{ij}(t)}$ to each of those messages. In this setting, the self-loop message from $i$ to itself will take weight $p_{ii} + \sum_{k=1}^n \Large{\mathds{1}}[w_{ik} + L_{ik}(t)=0] p_{ik}$ where the sum is over all neighboring nodes $k$ from which $i$ does not receive anything at iteration $t$. Define $\Phi(t) = \diag(\phi(t))$. We can determine which delay nodes at the ends of delay chains have information to be delivered by taking the product $J_{n \times b} \Phi(t-1)$ and locating which entries are $1$. Thus, to construct $\del{P}$ we locate all the entries equal to $1$ in matrix $J_{n \times b} \Phi(t-1)$ at row $i$ and columns corresponding to deliveries from $j$, and replace them by $\frac{p_{ij}}{\del{w}_{ij} + L_{ij}(t)}$. If $L_{ij}(t) = 1$ we also need to replace that entry with $\frac{p_{ij}}{\del{w}_{ij} + L_{ij}(t)}$. With a slight abuse of notation let us describe with $\bar{P}[L(t)]$ and $\bar{P}[\phi(t-1)]$ the operators that  replace the $1$s  in $L(t)$ and $J_{n \times b} \Phi(t-1)$ respectively with weights using $P$. If node $i$ receives no messages from neighbour $j$, then $\del{w}_{ij} + L_{ij}(t) = 0$ and we transfer the weight $p_{ij}$ to the self-loop message of $i$. The transition matrix $\del{P}(t)$ is now written as
\begin{align}
&\del{P}(t) =  \begin{bmatrix} \label{eq:qt_general}
	\del{P}_{1,1}(t) &  \bar{P}[\phi(t-1)] 	\\
            R(t) & C_{b \times b}
\end{bmatrix} \\
&\del{P}_{1,1}(t) =  I - \diag(\bar{P}[\phi(t-1)] \vc{1}_b + \bar{P}[L(t)] \vc{1}_n ) + \bar{P}[L(t)]. 
\end{align}
The upper left block of $\del{P}(t)$ has this form since for any row stochastic matrix $P$, we have $p_{ii} + \sum_{k=1}^n \Large{\mathds{1}}[\del{w}_{ik} + L_{ik}(t)=0] p_{ik} = 1 -  \sum_{k=1}^n \Large{\mathds{1}}[\del{w}_{ik} + L_{ik}(t) > 0] p_{ik}$ for each compute node $i$. This is just another way of saying that the portion of the weight not used on incoming messages at compute node $i$ from other neighbours is reassigned to the self loop message.

Observe that the rows of $\del{P}(t)$ either sum to zero or to one. Each row $i$ for $i \leq n$ (corresponding to a compute node) is stochastic by construction, while each row $i$ for $n < i \leq n+b$ (corresponding to a delay node) contains at most a single $1$ and all other elements are $0$. A row $i > n$ corresponding to a delay node $d_{1}^{r}$ will be a zero row if the compute node at the source of the corresponding edge did not send a message through the delay chain $r$. Let $\del{\vc{x}}(t) \in \mathds{R}^{n + b}$ denote the augmented state vector of compute and delay nodes. The consensus update equations using $\del{P}(t)$ are now
\begin{align}  \label{eq:randomdelay_consensus}
\del{\vc{x}}(t+1) = & \del{P}(t+1) \del{\vc{x}}(t), \ \ t \geq 0
\end{align}
where to construct $\del{P}(t+1)$ we need to first update the vector $\phi(t)$ according to \eqref{eq:phi}.

The presence of zero rows makes the transition matrices $\del{P}(t)$ not stochastic so we need a convergence proof specific to this family of matrices. As we see later, one advantage of Push-Sum consensus is that it simplifies the random delay model and we do not have this complication.

\subsection{Convergence under Random Delays}

We can show convergence of the random delay update model \eqref{eq:randomdelay_consensus} by inspecting the fundamental properties of the matrices $\{\del{P}(t) \colon t=1,2,\dots\}$. First we need two standard definitions \cite{contractingMatrices}:
\begin{dfn}
A square matrix $M$ is non-expansive with respect to a norm $\norm{\cdot}$ if for any vector $\vc{x}$, we have $\norm{M \vc{x}} \leq \norm{\vc{x}} $.
\end{dfn}
\begin{dfn}
A square matrix $M$ is paracontracting with respect to a norm $\norm{\cdot}$ if for any vector $\vc{x}$, we have $\norm{M \vc{x}} < \norm{\vc{x}}$ whenever $M \vc{x} \neq \vc{x}$.
\end{dfn}

From the construction of the random delay matrices, it is easy to see that the graphs represented by the adjacency matrices $A(t)$ are all connected, and in addition, every compute node performs an averaging operation of the incoming messages. We can thus show that the product of sufficiently many consecutive matrices $\del{P}(t)$ is a contractive mapping, leading to convergence.

\begin{thm} \label{thm:random_delay_convergence} The product $\del{P}_{2B+1}(t) = \prod_{s=0}^{2B} \del{P}(t+s)$ of $2B+1$ consecutive random delay matrices is non-expansive with respect to the infinity norms $\norm{\cdot}_{\infty}$ and $\norm{\cdot}_{-\infty}$. Moreover, for some integer $r \geq 1$ that depends on the network topology, the product  $\del{P}_{r(2B+1)}(t) $ is paracontracting. As a result, every non-empty node $i$ such that $1 \leq i \leq n+b$  and $\phi_i(t) > 0$  converges almost surely to the same value; i.e. $\del{x}_i(t) \rightarrow  v$  as $t \rightarrow \infty$.
\end{thm}

\begin{proof} Consider the linear random delayed consensus updates subsampled at intervals of $2B+1$ iterations:
\begin{equation}
\del{\vc{x}}(t) = \del{P}_{2B+1}(t) \del{\vc{x}}(t-1), \ \ t=1,2,\ldots
\end{equation}
Recall that in parallel to $\del{\vc{x}}(t)$ we have to evolve the vector $\vc{\phi}(t)$ which indicates which delay nodes are empty. To focus on the non-empty nodes,  define the vector $\vc{y}(t)$ such that $y_i(t) = \del{x}_(t)$ if $\phi_i(t) > 0$ and $y_i(t) = - \infty$ if $\phi_i(t) = 0$.

Let us  observe that the maximum value of $\vc{y}(t)$ is either equal to or smaller than the maximum value of $\vc{y}(t-1)$. If a compute node $i \leq n$ holds the maximum value of $\vc{y}(t-1)$, in $B+1$ iterations it is certain that $i$ will receive a message from a neighbouring compute node $j \leq n$. If at least one neighbour of $i$ has a smaller value than $i$, then the value of $i$ will be reduced because $i$ will set its new value to a convex combination of the more than one incoming messages (including the self message). However, $i$ may send its (maximum) value to a node $k \leq n$ through the delay chain of length $B$ at iteration $t$. Regardless of whether the value at $i$ is reduced or not, the maximum of $\vc{y}(t-1)$ will not change while it is traversing the delay chain towards $k$. When the message reaches $k$, node $k$'s value will be reduced unless all of its neighbours have sent messages to $k$ equal to the maximum. To summarize, the maximum value of $\vc{y}(t-1)$ after $2B+1$ iterations will either stay the same or be reduced. The maximum value will not change if multiple nodes  hold that value and there exist at least one node with no neighbours that contain a smaller value. As a result, the maximum value of the state vector will certainly be reduced after $r(2B+1)$ where $r=1,2,\ldots$ is defined as follows. Assume a node $i$ holds the maximum value of $\vc{y}(t-1)$. If at least one neighbour of $i$ holds a smaller value, then $r=1$. If all nodes in the distance $1$ neighbourhood $N^{1}(i)$ of $i$ also contain the maximum value then $r=2$. If the neighbours of the neighbours $N^{2}(i) = N^{1}(N^{1}(i))$ of $i$ contain the maximum value then $r=3$ and so on. Notice also that if  the delay nodes were real nodes initialized with random values such that a delay node contained the maximum value in $\vc{y}(t-1)$, then that value would reach a compute node and would be reduced via an averaging update in at most $B+1$ iterations. We have shown that $\del{P}_{2B+1}(t)$ is non-expansive with respect to $\norm{\cdot}_{\infty}$. Similarly, since averaging a set of numbers increases the smallest number in the set, $\del{P}_{2B+1}(t)$ is also non-expansive with respect to $\norm{\cdot}_{-\infty}$ if we define $\vc{y}'(t)$ so that $y_i'(t) = + \infty$ if $\phi_i(t) = 0$. Moreover, for a given network, we have shown that there exists an integer $r$ such that $\del{P}_{r(2B+1)}(t)$ certainly reduces the maximum value of $\vc{y}(t-1)$ and increases the minimum value of $\vc{y}'(t-1)$. In other words, every product $\del{P}_{r(2B+1)}(t)$  is paracontracting and thus every $r(2B+1)$ iterations the minimum and maximum values in the graph come close together and thus must converge to the same limit $v \in \mathds{R}$.
\end{proof}

Even though Theorem \ref{thm:random_delay_convergence} establishes convergence to consensus under random delays, the actual consensus value $v$ is difficult to characterize since it depends on the specific realization of the process---i.e., on the random matrices $\del{P}(t)$ used at every iteration. As future work, it might be possible to extend the results of~\cite{AsymptoticConsensusVal} to describe the statistics of $v$, however the extension is non-trivial since their results are based on the assumption that all the involved matrices do not have zeros in the diagonal which is not the case in our model. Here, we show that, as one might expect, $v$ is a convex combination of the initial conditions. We achieve this by showing that the top left $n \times n$ submatrix of $\del{P}(t)$ is a row stochastic matrix for all $t$.

After $t+1$ steps we have
\begin{align}
\del{\vc{x}}(t+1) = \del{P}(t+1) \del{P}(t) \cdots \del{P}(1) \del{\vc{x}}(0).
\end{align}
The product $\prod_{k=1}^{t} \del{P}(k)$ is a matrix with block structure
\begin{equation}
\prod_{k=1}^{t} \del{P}(k) = M(t) = \begin{bmatrix} \label{eq:definition_of_Mt}
	   M_{1}(t) & M_{3}(t)   \\[0.3em]
           M_{2}(t) & M_{4}(t)
\end{bmatrix} 
\end{equation}
where matrix $M_{1}(t)$ is  $n \times n$  and $M_{2}(t)$ is $b \times n$. So we have
\begin{align}
\del{\vc{x}}(t+1) =  &\del{P}(t+1) M(t)  \del{\vc{x}}(0) \notag \\
= &\begin{bmatrix}
	  \del{P}_{1,1}(t+1) & \bar{P}[\phi(t)]     \\
           R(t+1) & C_{b \times b}
\end{bmatrix}
 \begin{bmatrix}
	   M_{1}(t) & M_{3}(t)   \\[0.3em]
           M_{2}(t) & M_{4}(t)
 \end{bmatrix} \del{\vc{x}}(0).
\end{align}
From the last equation, we obtain two recursions 
\begin{align}
M_{1}(t+1) = & \Big(  I_{n \times n} - \diag\big(\bar{P}[\phi(t)] \vc{1}_b + \bar{P}[L(t+1)] \vc{1}_n\big) \notag \\
 & +\bar{P}[L(t+1)] \Big) M_{1}(t) + \bar{P}[\phi(t)] M_{2}(t) \label{eq:M1}  \\ 
M_{2}(t+1) = & R(t+1) M_{1}(t) + C_{b \times b} M_{2}(t). \label{eq:M2}
\end{align}
We will show that  $M_{1}(t)$ is row stochastic for all $t$ and that it converges to a rank-$1$ matrix. We begin by proving some intermediate lemmas and then proceed with the proof of the main theorem.

\begin{lemma} \label{lem:match_phi_M2}
For all $t$,  $M_{2}(t)$ and $\phi(t)$ have non-zero rows in exactly the same positions.
\end{lemma}

\begin{proof}
We will proceed inductively, using the expressions for how $M_{2}(t)$ and $\phi(t)$ evolve. We have $\phi(1) = R(1) \mathbf{1}_n  + C_{b \times b} \phi(0)  =  R(1) \mathbf{1}_n$ and $M_{2}(1) = R(1)$ so clearly the non-zero rows of $R(1)$ are the non-zero rows of $M_{2}(1)$, and they also result in non-zero entries of $\phi(1)$. For the inductive step, let us assume that $\phi(t)$ and $M_{2}(t)$ have non-zero rows in the same positions. At step $t+1$ we have $\phi(t+1) = R(t+1)\mathbf{1}_n + C_{b \times b} \phi(t)$ and $M_{2}(t+1) = R(t+1) M_{1}(t)  + C_{b \times b} M_{2}(t) $. If row $i$ of $\phi(t)$ and $M_{2}(t)$ is non-zero, then due to multiplication by the shift matrix $C_{b \times b}$, row $i+1$ of $\phi(t+1)$ and $M_{2}(t+1)$ will be non-zero. Moreover, if a row $i$ of $R(t+1)$ is non-zero then obviously row $i$ of $\phi(t+1)$ will be non-zero. For $M_{2}(t+1)$, we look at the term $R(t+1) M_{1}(t) $. Observe that  $M_{1}(t)$ has non-zero diagonal entries for all $t$. This is easy to see by the update equation \eqref{eq:M1} for $M_{1}(t)$. As a result, the product $R(t+1) M_{1}(t) $ will yield non-zero rows of $M_{2}(t+1)$ wherever a row of $R(t+1)$ is non-zero. This completes the inductive step of the proof.
\end{proof}

The next two lemmas are also inductive, and they are coupled in the sense that their proofs use each other's inductive hypothesis. Specifically, assuming that $M_{1}(t)$ is row stochastic and the non-zeros rows of $M_{2}(t)$ sum to $1$, we show that the non-zeros rows of $M_{2}(t+1)$ sum to $1$ and $M_{1}(t+1)$ is row stochastic respectively, establishing that both properties are true for all $t$. 

\begin{lemma} \label{lem:M2_stoch}
The non-zero rows of $M_{2}(t)$ sum to $1$ for all $t$.
\end{lemma}

\begin{proof}
Initially, $M_{2}(1) = R(1)$, and the base case is true. Suppose for every non-zero row $1 \leq i \leq b$ of $M_{2}(t)$ that $\sum_{j=1}^{n} [M_{2}(t)]_{ij} = 1$. Also by inductive hypothesis, suppose that $M_{1}(t)$ is row stochastic. We will show that the non-zero rows of $M_{2}(t+1)$ sum to $1$. Take any row $1 \leq i \leq b$ of $M_{2}(t+1)$. We have
\begin{align}
\sum_{j=1}^{n} [M_{2}(t+1)]_{ij} = & \sum_{j=1}^{n} [ R(t+1) M_{1}(t)  + C_{b \times b} M_{2}(t) ]_{ij} \notag \\ 
& =  \sum_{j=1}^{n}  [R(t+1) M_{1}(t) ]_{ij} + \sum_{j=1}^{n} [C_{b \times b} M_{2}(t) ]_{ij}.
\end{align}
Given the way the delay nodes are arranged in the random delay model, row $i$  of $R(t+1)$  corresponds to a delay node $d_{r_{1}}^{r_{2}}$ such that $1 \leq r_{2} \leq B$ and $r_{1} \leq r_{2}$. By definition, row $i$ of $R(t+1)$ will be zero if $r_{1} > 1$ and may be non-zero if $r_{1}=1$. We thus distinguish two cases:

$\bullet$ \textit{Case} $r_{1} = 1:$ By definition all rows of $C_{b \times b}$ corresponding to delay nodes at the beginning of delay chains (identified as $d_1^{r_2}$), are zero. If row $i= d_{1}^{r_{2}}$ of $R(t+1)$ is non-zero, it will have all entries equal to zero except  one entry equal to $1$ at some position $1 \leq q \leq n$. As a result
\begin{align}
\sum_{j=1}^{n} [M_{2}(t+1)]_{ij} = & \sum_{j=1}^{n} [R(t+1) M_{1}(t) ]_{ij} + \sum_{j=1}^{n} [C_{b \times b} M_{2}(t)]_{ij} \notag \\
=   \sum_{j=1}^{n}  [M_{1}(t)]_{qj}&  +  \sum_{j=1}^n \mathbf{0}_b^T [M_2(t)]_{:,j} =  \sum_{j=1}^{n} [M_{1}(t)]_{qj} = 1, 
\end{align}
since, by inductive hypothesis, $M_{1}(t)$ has stochastic rows. Of course, if row $i$ of $R(t+1)$ happens to contain only zeros, then the $i$-th row of $M_2(t+1)$ will be a zero row too. 

$\bullet$ \textit{Case} $r_{1} > 1:$ In this case $ \sum_{j=1}^{n} [R(t+1) M_{1}(t) ]_{ij} = 0$ and 
\begin{align}
\sum_{j=1}^{n} [M_{2}(t+1)]_{ij} = & \sum_{j=1}^{n} [C_{b \times b} M_{2}(t)]_{ij}.
\end{align}
Since $C_{b \times b}$ is just a shift matrix, each row $i > 1$ of $M_{2}(t+1)$ will equal to the row $i-1$ of $M_{2}(t)$ which by inductive hypothesis sums to $1$. The first row of $M_2(t+1)$ will be a zero row. 
\end{proof}

\begin{lemma} Matrix $M_1(t)$ is row stochastic. \label{thm:M1_stoch}
\end{lemma}

\begin{proof}
Proceeding inductively, the base case is true since $M_{1}(1) = I$. Assume at step $t > 1$ that $\sum_{j=1}^{n} [M_{1}(t)]_{ij} = 1$ for every row  $1 \leq i \leq n$. At step $t+1$ assume that compute node $i$ receives $\del{w}_{ij}$ messages from node $j$ through different delay chains plus possibly a message without delay if $L_{ij}(t+1) = 1$. Since the self loop message is always delivered without delay we know that $\del{w}_{ii}=1$ . We have
\begin{align}
\sum_{j=1}^n & [M_1(t+1)]_{ij} \notag \\
= &  \sum_{j=1}^n \Big[\Big(  I_{n \times n} - \diag(\bar{P}[\phi(t)] \vc{1}_b + \bar{P}[L(t+1)] \vc{1}_n) \notag \\
 & \ \ \ \ \ \ \ \ \ \ \ \ +\bar{P}[L(t+1)] \Big) M_{1}(t) + \bar{P}[\phi(t)] M_{2}(t)\Big]_{ij} \\
= & \underbrace{\sum_{j=1}^n \left[\Big(  I_{n \times n} - \diag(\bar{P}[\phi(t)] \vc{1}_b + \bar{P}[L(t+1)] \vc{1}_n)\Big) M_{1}(t)\right]_{ij}}_{T_1} \notag \\
& + \underbrace{\sum_{j=1}^n \big[ \bar{P}[L(t+1)] M_1(t)   + \bar{P}[\phi(t)] M_{2}(t) \big]_{ij}}_{T_2}.
\end{align}
Consider the term $T_1$ first, and notice that $I_{n \times n} - \diag(\bar{P}[\phi(t)] \vc{1}_b + \bar{P}[L(t+1)] \vc{1}_n)$ is a diagonal matrix so we have 
\begin{align}
 T_1 = &(1 - [\diag(\bar{P}[\phi(t)] \vc{1}_b +  \bar{P}[L(t+1)] \vc{1}_n ]_{ii}) \sum_{j=1}^n [M_1(t)]_{ij} \notag \\
  = &1 - \sum_{j=1}^n \mathds{1}[\del{w}_{ij} > 0 \text{\ or\ } L_{ij}(t+1) > 0] p_{ij}.
\end{align}
Next let us focus on term $T_2$ which is composed of two summands. For the first summand we have
\begin{align}
\sum_{j=1}^n & [\bar{P}[L(t+1)] M_1(t)]_{ij} \notag \\
= & \sum_{j=1}^n  \sum_{k=1}^n \bar{P}[L(t+1)]_{ik} [M_1(t)]_{kj}  \\
= & \sum_{k=1}^n \bar{P}[L(t+1)]_{ik} \sum_{j=1}^n [M_1(t)]_{kj} \\
= & \sum_{k=1}^n \bar{P}[L(t+1)]_{ik} \\
= & \sum_{k=1}^n L_{ik}(t+1) \frac{p_{ik}}{\del{w}_{ik} + L_{ik}(t+1)} \\
= & \sum_{j=1}^n \mathds{1} [L_{ij}(t+1)  > 0] L_{ij}(t+1) \frac{p_{ij}}{\del{w}_{ij} + L_{ij}(t+1)}.
\end{align}
To compute the second summand in $T_2$, from Lemma \ref{lem:match_phi_M2} we know that the non-zero rows of $M_2(t)$ are at the same position as those of $\phi(t)$. Observe now that those positions are the same as the non-zero rows of $J_{n \times b} \Phi(t)$ and thus the non-zero rows of $\bar{P}[\phi(t)]$. Assume that at iteration $t$ node $i$ receives delayed messages only from the compute nodes in the set $\mathcal{N}_i(t) \subseteq V$. Moreover, assume node $i$ receives $\del{w}_{i n_r} \geq 1$ messages from neighbour $n_r \in \mathcal{N}_i(t)$ through different delay chains. We have
\begin{align}
\sum_{j=1}^n [\bar{P}[\phi(t)] &M_2(t)]_{ij} = \sum_{j=1}^n  \bar{P}[\phi(t)]_{i,:} [M_2(t)]_{:,j}  \\
= & \sum_{j=1}^n  \sum_{n_r \in \mathcal{N}_i(t)}  \sum_{l=1}^{\del{w}_{in_r}}  \frac{p_{i n_r}}{\del{w}_{i n_r} + L_{i n_r}(t+1)}   [M_2(t)]_{n_rj}  \\
= &   \sum_{n_r \in \mathcal{N}_i(t)}  \sum_{l=1}^{\del{w}_{in_r}}  \frac{p_{i n_r}}{\del{w}_{i n_r} + L_{i n_r}(t+1)}   \sum_{j=1}^n [M_2(t)]_{n_rj}  \\
= &   \sum_{n_r \in \mathcal{N}_i(t)}   \frac{p_{i n_r}}{\del{w}_{i n_r} + L_{i n_r}(t+1)} \del{w}_{i n_r}  \\
= & \sum_{j=1}^n \mathds{1}[\del{w}_{ij} > 0] \frac{p_{ij}}{\del{w}_{i j} +  L_{i j}(t+1)} \del{w}_{ij}.
\end{align}
So now we see that
\begin{align}
T_2 = & \sum_{j=1}^n \mathds{1}[L_{ij}(t+1) > 0] L_{ij}(t+1) \frac{p_{ij}}{\del{w}_{ij} + L_{ij}(t+1)} \notag \\
& + \sum_{j=1}^n \mathds{1}[\del{w}_{ij} > 0] \frac{p_{ij}}{\del{w}_{i j} +  L_{i j}(t+1)} \del{w}_{ij} \\
= & \sum_{j=1}^n  \mathds{1}[\del{w}_{ij} > 0 \text{\ or \ } L_{ij}(t+1) > 0] \\
& \times \frac{p_{ij}}{\del{w}_{i j} +  L_{i j}(t+1)} (\del{w}_{ij} + L_{ij}(t+1)) \\
= & \sum_{j=1}^n  \mathds{1}[\del{w}_{ij} > 0\text{\ or \ } L_{ij}(t+1) > 0] p_{ij},
\end{align}
and finally 
\begin{align}
 \sum_{j=1}^n [M_1(t+1)]_{ij} = T_1 + T_2 = 1.
\end{align}
Therefore $M_1(t)$ is row stochastic for all $t$.
\end{proof}

Finally, we can state the  result as follows.
\begin{cor} \label{cor:rowstochastic_convergence}
Given a graph $G$ and a row stochastic consensus protocol $P$, if we run consensus on $G$ with random delays up to $B$ using updates \eqref{eq:randomdelay_consensus} with $\del{P}(t)$ given by \eqref{eq:qt_general}, all compute nodes of $G$ asymptotically reach consensus on a value $v$ that is a convex combination of their initial values.
\end{cor}
\begin{proof}
After $t$ iterations we have $\del{\vc{x}}(t) = M(t) \del{\vc{x}}(0)$  where $\vc{\del{x}}(t)$ is the augmented vector containing the values of the compute nodes followed by all the delay nodes. The delay nodes do not initially contain any information, so we have $[\del{\vc{x}}(0)]_{n+1:n+b} = 0$. After $t$ iterations,
\begin{align}
\del{x}_i(t) =  &M_1(t) [\del{\vc{x}}(0)]_{1:n} + M_3(t) [\del{\vc{x}}(0)]_{n+1:n+b}\\
 = & M_1(t) [\vc{x}(0)]_{1:n}.
\end{align}
So, as $t \rightarrow \infty$, since $\del{x}_i(t) \rightarrow v$ and $M_1(t)$ is row stochastic, $v$ is a convex combination of the initial values.
\end{proof}

As a last comment, notice the we achieve consensus on the compute nodes, even though the overall matrix $M(t)$ does not have a limit. Specifically, the rows corresponding to delay nodes oscillate between zero and non-zero values. However this  does not affect the sub matrix corresponding to the compute nodes. Notice also, that from this analysis we cannot say anything concrete about the rate of convergence. A convergence rate bound in expectation could be obtained by applying the Poincar\'e technique from the previous section on $\mathds{E}[\del{P}(t)]$. Alternatively, it might be possible to derive a more accurate bound by analyzing the recursions \eqref{eq:M1}, \eqref{eq:M2}. After realizing that $C^B = 0$, $M_2(t)$ can be eliminated given enough past terms, and the evolution of $M_1(t)$ resembles that of the impulse response of a multivariate $AR(B)$ model.

\section{Push-Sum Consensus}

The previous section studies the behaviour of general consensus protocols using row stochastic matrices in the presence of fixed and random delays. In the random delay case  the model is a bit involved due to the fact that we need to keep track of which delay nodes are empty, and also a compute node does not know how many messages it will receive at each iteration. Moreover, the convergence proof needs to be tailored specifically to the model because the resulting matrices $\del{P}(t)$ are not row stochastic. Even more importantly, we do not have a statement characterizing the convergence rate and the limiting state is a convex combination of the initial values at each node which is not necessarily the average. In this section we study a different consensus algorithm called Push-Sum. As we explain, Push-Sum is a more natural algorithm for distributed averaging in networks with delay; it alleviates all the aforementioned complications, simplifies the delay models, and always converges to the true average.

A simple asynchronous version of Push-Sum is proposed and analyzed in \cite{PushSum} for complete graphs. In \cite{WeightedGossip} the algorithm is analyzed in its general form for any graph. The Push-Sum protocol makes use of column stochastic consensus matrices and each node $i$ maintains two values: a cumulative estimate of the sum $s_i(t)$ and a weight $w_i(t)$. The local estimate of the average at each iteration is the ratio $x_i(t) = \frac{s_i(t)}{w_i(t)}$. The algorithm is initialized by setting 
\begin{align} \label{eq:pushsuminit}
\vc{s}(0) & = \vc{x}(0) \text{\ \ and\ \ } \vc{w}(0)  = \vc{1}.
\end{align}
Given the topology of the (directed) network $G$, we use at each iteration a  column stochastic matrix $P(t)$ respecting $G$. At each iteration, node $j$ splits its total sum $s_{j}(t)$ and weight $w_{j}(t)$ into shares $\Big\lbrace S_j(i) = \big(p_{ij}(t) s_{j}(t), p_{ij}(t) w_{j}(t) \big), i \in V \Big\rbrace$ where $\sum_{i=1}^{n} p_{ij}(t) = 1$, and sends to each neighbour $i$ the corresponding share $S_j(i)$. Equation \eqref{eq:pushsum} shows the actions performed at each receiver; i.e., simply add up all the incoming shares. In vector form the state evolves as
\begin{align} 
\vc{s}(t) & = P(t) \vc{s}(t-1)  \text{\ \ and\ \ } \vc{w}(t)  = P(t) \vc{w}(t-1) \label{eq:pushsum} \\ 
\vc{x}(t) & = \frac{\vc{s}(t)}{\vc{w}(t)}, \label{eq:avg_estimate_pushsum}
\end{align}
where the division of $\vc{s}(t)$ and $\vc{w}(t)$ is element-wise. We can verify that the updates \eqref{eq:pushsum} satisfy a conservation of mass property in the sense that for all $t \ge 0$,
\begin{align} \label{eq:mass}
 \sum_{i=1}^n s_i(t) & = \sum_{i=1}^n x_i(0) = \vc{1}^T x(0) = n x_{ave}\\
 \sum_{i=1}^n w_i(t) & = n.
\end{align}

To see why Push-Sum converges to the true average even in the time-varying case, assume $P(t)$ are sampled i.i.d.~such that $\mathds{E}[P]$ is irreducible at each iteration. Then the sequence $\{P(t)\}_{t=1}^\infty$ is weakly ergodic (Lemma $4.2$ in \cite{WeightedGossip}). Let us call $P^{\infty}$ the limit of the forward product $P(1)^T P(2)^T \cdots P(t)^T$ as $t \rightarrow \infty$. As a product of row stochastic matrices, $P^{\infty}$ is row stochastic with all rows the same. At any node $i$ we have
\begin{align} \label{eq:push_sum_convergence}
x_{i}(\infty)^T & = \frac{\big[\vc{s}(0)^T P^{\infty}  \big]_{i}}{\big[ \vc{w}(0)^T P^{\infty} \big]_{i}} = \frac{\big[ \vc{x}(0)^T  P^{\infty}  \big]_{i}}{\big[ \mathbf{1}^T  P^{\infty} \big]_{i}} = \frac{ \sum_{j=1}^{n} p_{ji}^{\infty} x_{j}(0) }{ \sum_{j=1}^{n} p_{ji}^{\infty}} \\
& = \frac{p_{1i}^{\infty}  \sum_{j=1}^{n}  x_{j}(0) }{ p_{1i}^{\infty} \sum_{j=1}^{n} 1} = \frac{ \sum_{j=1}^{n}  x_{j}(0) }{n} = x_{ave}.
\end{align}
We use the fact that all rows of $P^{\infty}$ are the same; i.e. $P_{ji}^\infty = P_{1i}^\infty$, for all $i,j$. For a formal proof see \cite{WeightedGossip}. Notice that Push-Sum computes the average without using doubly stochastic matrices or requiring knowledge of the stationary distribution a priori. 

\subsection{Consensus with Fixed Delays using Push-Sum}

In the case of fixed delays, the construction of a protocol with delays $\del{P}$ based on an initial protocol $P$ is the same as in Section \ref{sec:fixed}. The only difference is that we start with a column stochastic protocol $P$ and convert it to a new column stochastic matrix $\del{P}$ by adding delays one edge at a time. For example, if we start with the protocol  \eqref{eq:examplePorig}, after adding a delay of $2$  on the edge $(1,2)$ we have
\begin{align}
\del{P} = \kbordermatrix{
~ & 1 & 2 & 3 & d_{1} & d_{2} \cr
1 &  \frac{2}{3} & \frac{1}{3} & 0         & 0 & 0   \cr
2 &  0 & \frac{1}{3} & \frac{1}{2} & 0 & 1   \cr
3 &    \frac{1}{6}          & \frac{1}{3} & \frac{1}{2} & 0 & 0   \cr
d_{1} &       \frac{1}{6}           & 0                 & 0           & 0 & 0   \cr
d_{2} &       0           & 0                 & 0           & 1 & 0  \cr
}.
\end{align}
In the case of Push-Sum, delay node $d_1$ receives $\frac{1}{6}$ of the share of node $1$. Using  $\del{P}$, average consensus is achieved by iterating
\begin{align}
\del{\vc{s}}(t) = \del{P} \del{\vc{s}}(t-1),\  \del{\vc{w}}(t) = \del{P} \del{\vc{w}}(t-1).
\end{align}
For the purpose of analysis, we initialize the delay nodes with $s_i(0) = w_i(0)=0, n+1 \leq i \leq n+b$, or in vector form,
\begin{align} \label{eq:pushsuminit}
\del{\vc{s}}(0) & = [\vc{x}(0)^{T} \ \ \ \vc{0}_{b}^{T}]^{T} \\
\del{\vc{w}}(0) & = [\vc{1}_{n}^{T} \ \ \ \ \ \ \ \vc{0}_{b}^{T}]^{T}.
\end{align}
If we run Push-Sum using the delayed consensus protocol $\del{P}$, writing $\del{P}^{\infty}$ for the limit of $\del{P}^{t}$ as $t \rightarrow \infty$ we see that the estimate of the average $\vc{x}_i$ at each node $i$ will be the true average of the initial values:
\begin{align}
x_{i}(\infty) & = \frac{\big[\del{P}^{\infty} \del{\vc{s}}(0) \big]_{i}}{\big[\del{P}^{\infty}  \del{\vc{w}}(0) \big]_{i}} = \frac{\big[\del{P}^{\infty} [\vc{x}(0)^{T}\ \ \vc{0}_{b}^{T}]^{T} \big]_{i}}{\big[\del{P}^{\infty}  [\vc{1}_{n}^{T} \ \ \vc{0}_{b}^{T}]^{T} \big]_{i}} \\
 = &\frac{ \sum_{j=1}^{n} \del{P}_{ij}^{\infty} x_{j}(0) }{ \sum_{j=1}^{n} \del{P}_{ij}^{\infty}}  = \frac{\del{P}_{i1}^{\infty}  \sum_{j=1}^{n}  x_{j}(0) }{ \del{P}_{i1}^{\infty} \sum_{j=1}^{n} 1} = \frac{ \sum_{j=1}^{n}  x_{j}(0) }{n}
\end{align}
since $\del{P}$ is column stochastic and $\del{P}^\infty$ will have identical columns. Obviously, the convergence rate bound \eqref{eq:aq} applies here as well.

\subsection{Consensus with Random Delays using Push-Sum} In row stochastic protocols with random delays, we need an indicator vector $\phi(t)$ to know whether a delay node contains information or is empty. We also need to assign the portion of the weight that is being unused to the self-loop message. Both of those complications arise from the fact that we do not know how many messages will be received at each iteration. With Push-Sum consensus however, the semantics suggest that the sending node decides how much weight to assign to each outgoing message, and each receiving node simply sums up the incoming $s$ and $w$ values without caring about the number of incoming messages. This fact simplifies both the model and the convergence analysis when we account for time-varying delays. 

Recall from the random delay model construction that the adjacency matrix $A(t)$ is given by \eqref{eq:randomdelayA}. However, now we are given a column stochastic matrix $P$ and need to construct a column stochastic matrix $\del{P}(t)$. Since $P$ indicates the outgoing weights,  the construction is straightforward:
\begin{equation}
\del{P}(t) = \begin{bmatrix} \label{eq:qt_general_pushsum}
	   \operatorname{diag}(P) + P \circ L(t) & J_{n \times b} \\
            \bar{P}[R(t)] & C_{b \times b}
\end{bmatrix},
\end{equation}
where, by $\operatorname{diag}(P)$ we mean a matrix with diagonal entries the same as those of $P$ and off-diagonal entries set to zero, and where $\circ$ denotes entry-wise (Hadamard) matrix multiplication. We define the operator $\bar{P}[R(t)]$ a bit differently than in the previous section. If $[R(t)]_{d_1^r,j} = 1$, where $d_1^r$ is the first node on a delay chain from compute node $j$ to compute node $i$, then we set $\bar{P}[R(t)]_{d_1^r,j} = p_{ij}$. Again for the purpose of analysis we initialize the $s$ and $w$ values for the delay nodes to zero. 

With Push-Sum, the model is simplified because we no longer need the vector $\phi$ to indicate which delay nodes contain information. The reason is that we have the weights $\vc{w}$ and an empty delay node is represented by having a weight of zero. Notice, in addition, that $\del{P}(t)$ is column stochastic by construction and does not contain zero columns. This allows us to use  weak ergodicity theory  \cite{SenetaBook,WeakErgodicityConditionsJadbabaie07} to establish convergence. 

\subsection{Convergence of Push-Sum consensus with Random Delays}

Using the random delay model with column stochastic matrices yields a forward product, and to prove convergence of this algorithm we need to establish weak ergodicity as was mentioned at the end of Section \ref{sec:consensus}. Since each matrix $\del{P}(t)$ in \eqref{eq:qt_general_pushsum} contains zeros on the diagonal, we cannot apply known results directly. In this section we derive a worst case (pessimistic) geometric convergence rate. We first need the following lemma.

\begin{lemma} \label{lem:delay_diameter}
If a strongly connected graph $G$ has diameter $D$, the graph $\del{G}$ obtained by adding arbitrary delays of up to $B$ on each edge has diameter at most $\del{D} \leq (B+1) D + B+1$.
\end{lemma}
\begin{proof}
Let $K = v \rightarrow v_{1} \rightarrow \cdots \rightarrow v_{D-1} \rightarrow w$ be a path in $G$ with length equal to $D$. By adding at most $B$ delay nodes per directed edge, each edge of $G$ is replaced by $B+1$ edges in $\del{G}$ and the corresponding path $\del{K}$ has length $(B+1) D$ in $\del{G}$. All neighbours of $v$ and $w$ in $G$ belong to $K$ or else the diameter would be longer. Suppose that in the worst case, $v$ has a neighbor $z_{1} \ne v_1$ and $w$ has a neighbor $z_{2} \ne v_{D-1}$ in $G$. After adding delays, the longest path in $\del{G}$ goes from the delay node in the middle of the longest delay chain between $z_{1}$ and $v$ and the delay node in the middle of the longest delay chain between $z_2$ and $w$ and has length at most $\del{D} \leq (B+1)D + \frac{B+1}{2} + \frac{B+1}{2} = (B+1)D + B + 1$.
\end{proof}

Now we can state the main convergence result of this section.
\begin{thm}
If we run Push-Sum on a strongly connected graph $G$ using a column stochastic protocol $P$, then in the presence of bounded time-varying delays modelled by \eqref{eq:qt_general_pushsum}, average consensus is achieved at a geometric rate. 
\end{thm}

\sloppypar \begin{proof}
Since $G$ is strongly connected, due to the way we model random delays, at each instant $t$ there exists a path between any two compute nodes $i$ and $j$. As a consequence, due to Lemma \ref{lem:delay_diameter}, every column $j \leq n$ of every sub-product matrix $F(r, r + \del{D}) = \del{P}(r)^T \del{P}(r+1)^T \cdots \del{P}(r + \del{D})^T$ contains positive entries\footnote{In other words after $\del{D}$ iterations every compute node communicates with every other compute node.}. This means that for the (improper) coefficient of ergodicity $c(\cdot)$ \cite{SenetaBook}(p. $137$)
\begin{align}
c\big(F(r, r + \del{D})\big)  = &1 - \max_{1 \leq s \leq n+b}(\min_k [F(r, r + \del{D})]_{ks}) \\
 \leq  &1 - \max_{1 \leq s \leq n}(\min_k [F(r, r + \del{D})]_{ks})  < 1
\end{align}
since the maximum over the minimum values in the compute node columns is certainly not zero. Now, if we run consensus with random delays for $t > \del{D}$ steps we divide the forward product $F(1,t)$ into
\begin{align}
F(1,t) = & \prod_{k=1}^{\frac{t}{\del{D}}} F((k-1)\del{D} +1, k \del{D}) \\
 = &F(1,\del{D}) F(\del{D}+1, 2 \del{D}) \cdots F(t-\del{D}+1, t)
\end{align}
and as explained above, $c\big(F((k-1)\del{D} +1, k \del{D}) \big) < 1$ for each term. Now immediately we see that $\sum_{k=1}^\infty \left[1 - c\big(F((k-1)\del{D} +1, k \del{D}) \big) \right] = \infty$, and from Theorem $4.9$ in \cite{SenetaBook}, the product $F(1,t)$ is weakly ergodic. Based on a derivation similar to \eqref{eq:push_sum_convergence}, after initializing the $s$ and $w$ values of the delay nodes to zero,  Push-Sum converges to the true average. Furthermore, if $\max_k\big(F((k-1)\del{D} +1, k \del{D}) \big) \leq c_0 < 1$, the forward product converges geometrically at a rate no worse than $c_0$. 
\end{proof}

\section{Simulations}
\label{sec:simulations}

In this section we use simulations to illustrate the important concepts discussed so far. The first experiment verifies Theorem \ref{thm:poincare} and Corollary \ref{cor:poincare}. One difficulty with verifying these results numerically is that  Theorem \ref{thm:poincare} describes the effect of fixed delays relative to a consensus protocol $P$ on a graph $G$ without delays. To compute the Poincar\'e constant $\del{K}$ explicitly we still need to find a set of canonical paths in $G$ and apply \eqref{eq:poincareK} which can be tedious. Instead, we estimate $\del{K}$ as follows. For a given network of $15$ nodes, protocol $P$ and delay bound $B$, we randomly select delays for all edges, construct $U(\del{P}_{lazy})$ as explained in Section~\ref{sec:fixed} and compute the second eigenvalue of $U$. For each bound $B$ we repeat this procedure $50$ times. Since $\del{K} \geq \frac{1}{1 - \lambda_2(U(\del{P}_{lazy}))}$ we keep the largest $\lambda_2$ out of the $50$ trials to approximately maximize the lower bound on $\del{K}$. Figure \ref{fig:poincarebound} illustrates that the inverse spectral gap increases almost quadratically with $B$. It appears that $O(B^2)$ might be increasing faster than $\del{K}$ so our bound might be loose but not dramatically so. The mismatch could also be a result of poor approximation on $\del{K}$ since for larger $B$, $50$ trials might not be enough to capture the worst possible scenario.

\begin{figure}
\begin{center}
\includegraphics[width=3.5in]{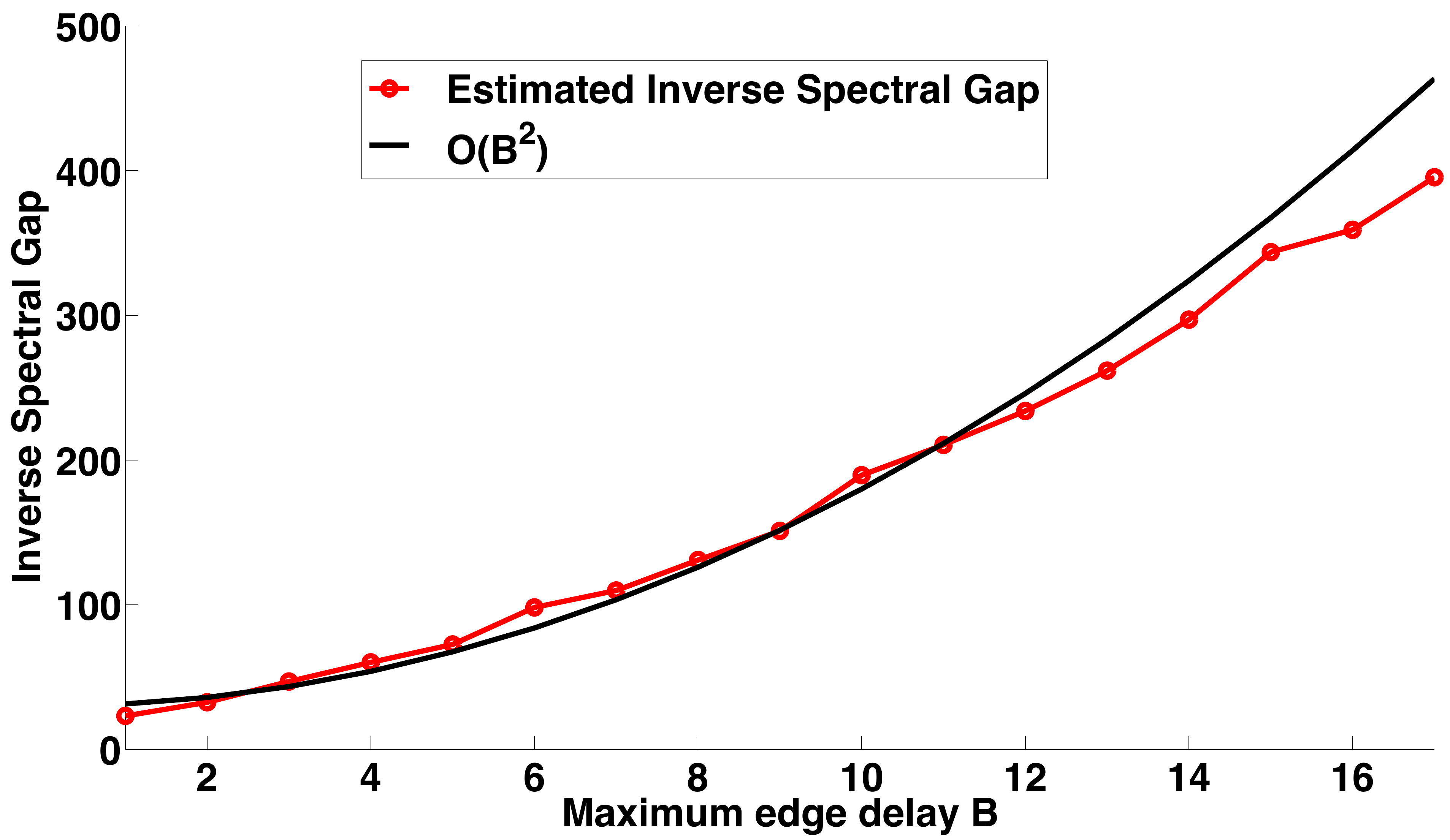}
\end{center}
\caption{\label{fig:poincarebound} (Red) Estimated inverse spectral gap $\frac{1}{1 - \lambda_2(U(\del{P}_{lazy}))}$ for a network $G$ of $15$ nodes when increasing the upper bound $B$ of fixed delays. Each data point is the maximum over $50$ randomly selected delay distributions over the edges of $G$. (Black) An approximate fit of an $O(B^2)$ curve to show that the inverse spectral gap does not deteriorate by worse than a quadratic factor as we increase $B$.}
\end{figure}

In a second simulation we investigate the case of time-varying delays. For a network with $5$ nodes and a maximum random delay of $B=5$, we plot the evolution of the node values when running consensus with equation \eqref{eq:randomdelay_consensus} and Push-Sum using consensus matrices of the form \eqref{eq:qt_general_pushsum}. We initialize the node values to be the node ids $1$ through $5$. In both cases we start with a random row stochastic protocol $P$ without delays and use its transpose to generate the Push-Sum weights. Figure \ref{fig:randomdelay_consensus} illustrates that since $P$ is not doubly stochastic, the compute nodes reach consensus as Corollary \ref{cor:rowstochastic_convergence} suggests, but the consensus value is not the average. Even worse, if we run the simulation again, the different random delays at each iteration will yield a different consensus value. With Push-Sum, on the other hand, the compute nodes always converge to the true average.

\begin{figure}
\begin{center}
\includegraphics[width=3.5in]{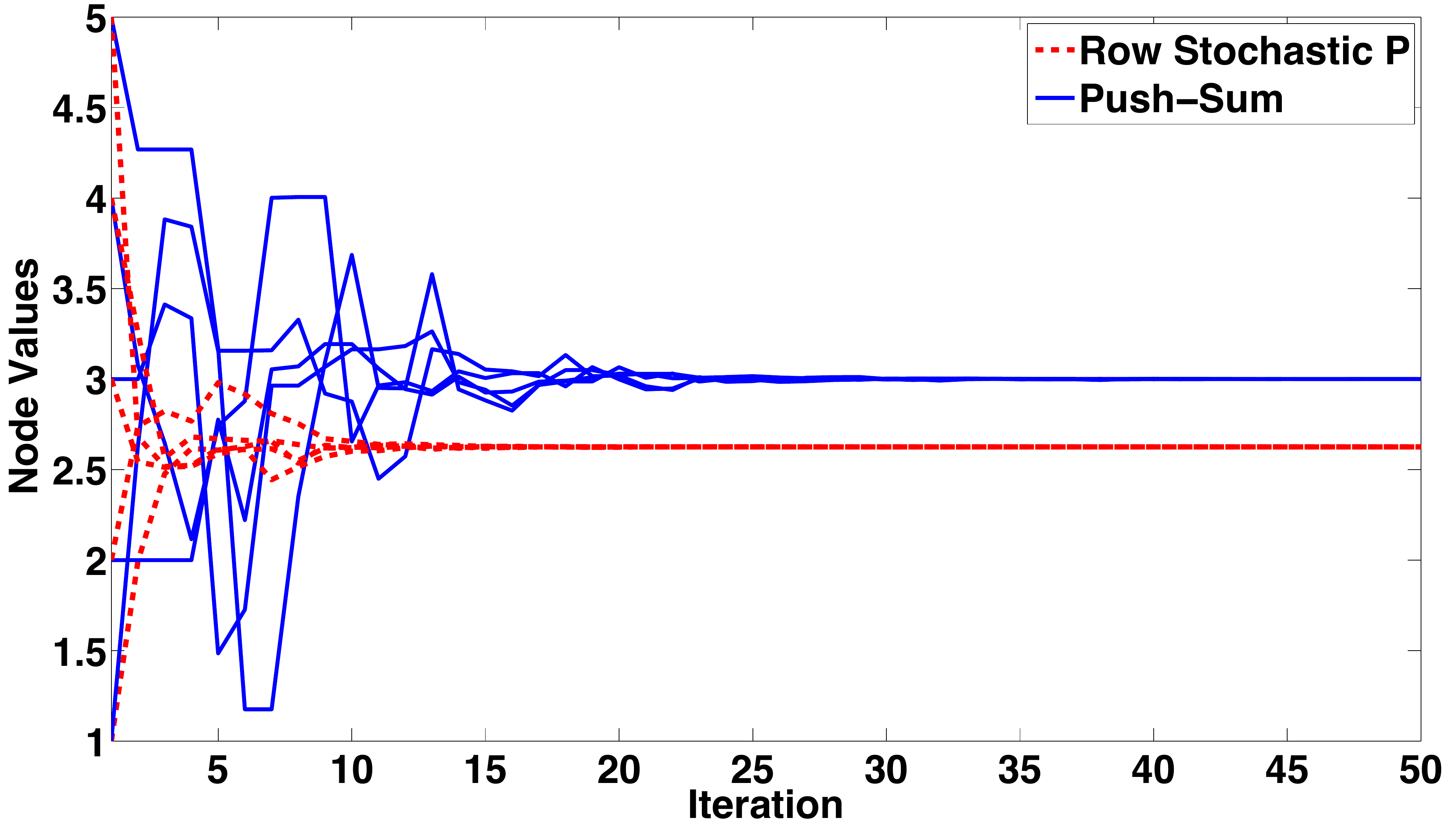}
\end{center}
\caption{\label{fig:randomdelay_consensus} Evolution of the node values on a graph of $5$ nodes with random delays no more than $B=5$. The true average is $x_{ave} = 3$. (Blue) With Push-Sum all nodes reach consensus to the correct average. (Red) Using a row stochastic matrix, as expected consensus is reached but not to the average and the consensus value varies between executions.}
\end{figure}

\section{Concluding Remarks and Future Work}
\label{sec:future}

In this paper we analyze the effect of communication delays in distributed algorithms for consensus and averaging. Initially we assume that each directed link of a communication network $G$ delivers messages with some fixed delay $B$. Delays on different links need not be equal. We show how to model the effect of delays by augmenting $G$ with artificial delay nodes and then use geometric arguments to show that the inverse spectral gap of a consensus protocol $\del{P}$ in the presence of delays does not increase faster than $\Theta(B^2)$. Thus, we still have exponentially fast convergence to a value which in general is not the average. For fixed row stochastic protocols, we can achieve average consensus by rescaling the initial values as explained in Section \ref{sec:consensus}. 

Next, we show how to model time-varying delays---a scenario that is more realistic but also harder to analyze. For general row stochastic consensus protocols we show that convergence to consensus is still guaranteed although the consensus value is itself a random variable. In the last part of the paper we propose and analyze the use of a different consensus protocol based on column stochastic matrices called Push-Sum. With Push-Sum, convergence to the average is always guaranteed and the analysis of the time-varying delay model is significantly simplified. These facts are in agreement with \cite{TsianosCDC2012}, suggesting that Push-Sum is more suitable for practical implementations. 

In the future, for the fixed delay scenario we would like to investigate the following optimization problem: Given a network $G$ and the fixed delays on its links, what is the consensus protocol $P$ that respects the structure of $G$ and reaches consensus as fast as possible in the presence of fixed delays? Notice that since we can use Push-Sum, any column stochastic matrix that does not add edges to $G$ will compute the true average and we are looking for the matrix with the smallest second eigenvalue. It would be interesting to investigate if the techniques used for second eigenvalue optimization for symmetric protocols (see e.g., \cite{randomizedGossip}) could be extended to answer this question. 

At the same time, for our time-varying delay models, the analysis only guarantees convergence and a loose geometric bound in the case of Push-Sum. It would be useful to have a more precise characterization of the convergence rate and to extend the Poincar\'e technique presented in this paper to understand how much do time-varying delays slow down convergence.

\section{Appendix: Proof of Theorem \ref{thm:poincare}}


Consider a graph $G$ with a consensus protocol $P$. Given a set of canonical paths $\Gamma = \{\gamma_{xy}\}$ on $G$ we can compute the Poincar\'e constant $K$. If each link of $G$ delivers messages with some arbitrary fixed delay of no more than $B$, we will show that the Poincar\'e constant $\del{K}$ of $\del{G}$ using the lazy additive reversibilization $U$ of $\del{P}$ is bounded like $\del{K} \leq Z K$ where $Z = \Theta(B^2)$.

We start with the definition of the Poincar\'e constant for $\del{K}$ and use the path associations discussed already to break the sum over all paths into nine summands. Assume that there are $N_{vw}$ canonical paths in $G$ that go through the bottleneck edge $e=(v,w)$ of $G$ and let the bottleneck edge of $\del{G}$ be $\del{e} = (u,z)$ where $u$ is in the set $v^+$ and $z$ is in the the set $w^-$. Let $x,y$ denote the starting and ending node of a path $\del{\gamma}_{i}$. We have
\begin{align} \label{eq:kd_bound}
\del{K} = & \frac{1}{\del{\pi}_{v^+}  [U]_{uz} } \big( T_{1}[x \rightarrow y] +  T_2[x \rightarrow y^-] + T_3[x \rightarrow y^+] \notag \\
& + T_4[x^- \rightarrow y^-] + T_5[x^- \rightarrow y] + T_6[x^- \rightarrow y^+] \notag \\
& + T_7[x^+ \rightarrow y^-] +  T_8[x^+ \rightarrow y] + T_9[ x^+ \rightarrow y^+]\big)
\end{align}
with
\begin{align}
& T_1  =  \sum_{i=1, \del{e} \in \del{\gamma}_{i}}^{N^{vw}}  \abs{\del{\gamma}_{i}} \del{\pi}_x \del{\pi}_y \\
& T_2  =  \sum_{i=1, \del{e} \in \del{\gamma}_{i}}^{N^{vw}} \sum_{k=-\frac{B^{-}}{2}}^{-1} \big( \abs{\del{\gamma}_{i}} + k\big) \del{\pi}_x \del{\pi}_{y^{-}} \\
& T_3  =  \sum_{i=1, \del{e} \in \del{\gamma}_{i}}^{N^{vw}} \sum_{r=1}^{deg(y)} \sum_{k=1}^{\frac{B_{r}}{2}} \big( \abs{\del{\gamma}_{i}} + k\big) \del{\pi}_x \del{\pi}_{y^{+}_{r}} \\
& T_4  = \sum_{i=1, \del{e} \in \del{\gamma}_{i}}^{N^{vw}} \sum_{h=1}^{deg(x)} \sum_{j=-\frac{B_{h}}{2}}^{-1} \sum_{k=-\frac{B^{-}}{2}}^{-1} \big( \abs{\del{\gamma}_{i}} + j + k \big)  \del{\pi}_{x^{-}_{h}}  \del{\pi}_{y^{-}} \\
& T_5  = \sum_{i=1, \del{e} \in \del{\gamma}_{i}}^{N^{vw}} \sum_{h=1}^{deg(x)} \sum_{j=-\frac{B_{h}}{2}}^{-1} \big( \abs{\del{\gamma}_{i}} + j \big) \del{\pi}_{x^{-}_{h}}  \del{\pi}_y \\
& T_6  = \sum_{i=1, \del{e} \in \del{\gamma}_{i}}^{N^{vw}} \sum_{h=1}^{deg(x)} \sum_{r=1}^{deg(y)} \sum_{j=-\frac{B_{h}}{2}}^{-1} \sum_{k=1}^{\frac{B_{r}}{2}}    \big( \abs{\del{\gamma}_{i}} + j + k\big) \del{\pi}_{x^{-}_{h}}  \del{\pi}_{y^{+}_{r}} \\
& T_7  = \sum_{i=1, \del{e} \in \del{\gamma}_{i}}^{N^{vw}} \sum_{j=-\frac{B^{+}}{2}}^{-1} \sum_{k=-\frac{B^{-}}{2}}^{-1}  \big( \abs{ \del{\gamma}_{i}} + j + k\big) \del{\pi}_{x^{+}}  \del{\pi}_{y^{-}} \\
& T_8  = \sum_{i=1, \del{e} \in \del{\gamma}_{i}}^{N^{vw}} \sum_{j=-\frac{B^{+}}{2}}^{-1}  \big( \abs{\del{\gamma}_{i}} + j \big) \del{\pi}_{x^{+}}  \del{\pi}_y \\
& T_9  =\sum_{i=1, \del{e} \in \del{\gamma}_{i}}^{N^{vw}} \sum_{r=1}^{deg(y)} \sum_{j=-\frac{B^{+}}{2}}^{-1}  \sum_{k=1}^{\frac{B_{r}}{2}} \big( \abs{ \del{\gamma}_{i}} + j +k \big) \del{\pi}_{x^{+}} \del{\pi}_{y^{+}_{r}} 
\end{align}
To obtain a cleaner bound for $\del{K}$ we assume that $P$ is doubly stochastic, recalling that the stationary distribution of delay nodes is $\del{\pi}_{x^{*}} \leq p \del{\pi}_x = \frac{p \pi_x}{c}$ for  $p = \max_{i \neq j}(p_{ij})$ and replacing $*$ with either $+, -$. Recall also that each path in $\del{G}$ corresponds to exactly one path in $G$. Below we show how to bound the term $T_6$; bounds for all of the other terms defined above are obtained using similar arguments. Observe that for every path $\gamma_{xy}$ between compute nodes $x$ and $y$, if $\gamma_{xy}$ goes through a bottleneck edge $e$ in $G$, then all the delay paths $\del{\gamma}$ that are associated with $\gamma_{xy}$ will go through $\del{e}$ in the middle of the delay chain that replaces $e$. So, for term $T_6$ we have
\begin{align}
T_6 \leq & \sum_{i=1, e \in \gamma_{i}}^{N^{vw}} \sum_{h=1}^{\operatorname{deg}(x)} \sum_{r=1}^{\operatorname{deg}(y)} \sum_{j=-\frac{B_{h}}{2}}^{-1} \sum_{k=1}^{\frac{B_{r}}{2}}    \big( (B+1) \abs{\gamma_{i}} + j + k\big)  \notag \\
& \times \frac{ p \pi_x}{c}  \frac{p \pi_y}{c} \\
\leq & \frac{p^{2}}{c^2} \sum_{i=1, e \in \gamma_{i}}^{N^{vw}} \operatorname{deg}(x) \operatorname{deg}(y) \notag \\
& \times \sum_{j=-\frac{B}{2}}^{-1} \sum_{k=1}^{\frac{B}{2}}    \big( (B+1) \abs{\gamma_{i}} + j + k\big)  \pi_x \pi_y 
\end{align}
Now since all paths $\gamma_{i}$ are at least one edge long, bounding the node degrees by the maximum degree $d_{max}$ in $G$ gives
\begin{align}
T_6 \leq & \frac{p^{2}}{c^2} d_{max}^{2} \sum_{j=-\frac{B}{2}}^{-1} \sum_{k=1}^{\frac{B}{2}}   
\big( (B+1)  + j + k\big) \notag \\
& \times \sum_{i=1, e \in \gamma_{i}}^{N^{vw}}   \abs{\gamma_{i}} \pi_x \pi_y \\
= & \frac{p^{2} d_{max}^{2} }{c^2} \frac{B^3 + B^2}{4} \sum_{i=1, e \in \gamma_{i}}^{N^{vw}}   \abs{\gamma_{i}} \pi_x \pi_y 
\end{align}
Through a similar derivation, all nine terms can be bound by a constant times $ \sum_{i=1, e \in \gamma_{i}}^{N^{vw}}   \abs{\gamma_{i}} \pi_x \pi_y$ which appears in the expression for the Poincar\'e constant $K$ without delays (see \eqref{eq:poincareK}). To make the exact expression for $K$ appear, we focus on the leading term in \eqref{eq:kd_bound} to see that
\begin{align}
 \frac{1}{\del{\pi}_{v^+}  [U]_{uz} c^{2}} = & \frac{c}{\pi_v  \frac{[U]_{uz} + [\tilde{U}]_{uz}}{2} c^{2}} =  \frac{2}{\pi_v  ( [U]_{uz} + 0) c} \\
 = & \frac{2}{\pi_v  c}  = \frac{2 p_{vw}}{c} \frac{1}{\pi_v p_{vw}}.
\end{align}
Next, remembering that $e=(v,w)$ is the bottleneck edge, after computing the exact constants in all terms, we write $\del{K} \leq Z K$ where $Z$ is a function of the node degrees, edge delays and consensus matrix $P$. Specifically,
\begin{align}
\del{K} \leq &  \frac{2 p_{vw}}{c}    \Big[ (B+1) + p \frac{3B^2 + 2B}{8} + p\ d_{max} \frac{5B^2 + 6B}{8}  \notag \\
& +  p^{2} d_{max} \frac{B^3}{8} +  p d_{max} \frac{3B^2 + 2B}{8}  + p^{2} d_{max}^{2} \frac{B^3 + B^2}{4} \notag \\
& + p^{2} \frac{B^3}{8} +  p \frac{3B^2 + 2B}{8}  +  p^{2} d_{max} \frac{B^3 + B^2}{4} \Big] \notag \\
& \times   \underbrace{\frac{1}{\pi_v p_{vw}} \sum_{i=1, e \in \gamma_{i}}^{N^{vw}}   \abs{\gamma_{i}} \pi_x \pi_y}_{K}  =Z K.
\end{align}
Finally, focusing on the expression for $Z$, after some algebra, we see that
\begin{align}
Z =  \frac{ p_{vw}}{4 c} \Big[&p^{2}(2d_{max}^{2} + 3d_{max}+1) B^{3}  \notag \\
& + p(2pd_{max}^{2} + 2pd_{max} + 8d_{max} + 6) B^{2}  \notag \\
& + (8pd_{max} + p + 8) B + 8\Big]
\end{align}
which completes the proof.

\bibliographystyle{IEEEtran} 
\bibliography{../PhDThesis/References}

\begin{thebibliography}{10}
\providecommand{\url}[1]{#1}
\csname url@samestyle\endcsname
\providecommand{\newblock}{\relax}
\providecommand{\bibinfo}[2]{#2}
\providecommand{\BIBentrySTDinterwordspacing}{\spaceskip=0pt\relax}
\providecommand{\BIBentryALTinterwordstretchfactor}{4}
\providecommand{\BIBentryALTinterwordspacing}{\spaceskip=\fontdimen2\font plus
\BIBentryALTinterwordstretchfactor\fontdimen3\font minus
  \fontdimen4\font\relax}
\providecommand{\BIBforeignlanguage}[2]{{%
\expandafter\ifx\csname l@#1\endcsname\relax
\typeout{** WARNING: IEEEtran.bst: No hyphenation pattern has been}%
\typeout{** loaded for the language `#1'. Using the pattern for}%
\typeout{** the default language instead.}%
\else
\language=\csname l@#1\endcsname
\fi
#2}}
\providecommand{\BIBdecl}{\relax}
\BIBdecl

\bibitem{tsianosAllerton2011}
K.~I. Tsianos and M.~G. Rabbat, ``Distributed consensus and optimization under
  communication delays,'' in \emph{49th Allerton Conference on Communication,
  Control, and Computing}, 2011.

\bibitem{PushSum}
D.~Kempe, A.~Dobra, and J.~Gehrke, ``Gossip-based computation of aggregate
  information,'' in \emph{FOCS, vol. 44. IEEE Computer Society Press, pp.
  482--491}, 2003.

\bibitem{langfordBook}
R.~Bekkerman, M.~Bilenko, and J.~Langford, \emph{Scaling up Machine Learning,
  Parallel and Distributed Approaches}.\hskip 1em plus 0.5em minus 0.4em\relax
  Cambridge University Press, 2011.

\bibitem{boydAlternatingMultipliers}
S.~Boyd, N.~Parikh, E.~Chu, B.~Peleato, and J.~Eckstein, ``Distributed
  optimization and statistical learning via the alternating direction method of
  multipliers,'' \emph{Foundations and Trends in Machine Learning}, vol.~3,
  no.~1, pp. 1--122, 2010.

\bibitem{dualAveraging}
J.~Duchi, A.~Agarwal, and M.~Wainwright, ``Dual averaging for distributed
  optimization: Convergence analysis and network scaling,'' \emph{IEEE
  Transactions on Automatic Control}, vol.~57, no.~3, pp. 592--606, 2011.

\bibitem{distrStochSubgrOpt}
S.~S. Ram, A.~Nedic, and V.~V. Veeravalli, ``Distributed stochastic subgradient
  projection algorithms for convex optimization,'' \emph{Journal of
  Optimization Theory and Applications}, vol. 147, no.~3, pp. 516--545, 2011.

\bibitem{nedicDistributedOptimization}
A.~Nedic and A.~Ozdaglar, ``Distributed subgradient methods for multi-agent
  optimization,'' \emph{IEEE Transactions on Automatic Control}, vol.~54,
  no.~1, January 2009.

\bibitem{JohanssonIncrementalSubgrad}
B.~Johansson, M.~Rabi, and M.~Johansson, ``A randomized incremental subgradient
  method for distributed optimization in networked systems,'' \emph{SIAM
  Journal on Control and Optimization}, vol.~20, no.~3, 2009.

\bibitem{gossipReview}
A.~G. Dimakis, S.~Kar, J.~M. Moura, M.~G. Rabbat, and A.~Scaglione, ``Gossip
  algorithms for distributed signal processing,'' \emph{Proceedings of the
  IEEE}, vol.~98, no.~11, pp. 1847 -- 1864, November 2010.

\bibitem{saberReview}
R.~Olfati-Saber, J.~A. Fax, and R.~M. Murray, ``Consensus and cooperation in
  networked multi-agent systems,'' in \emph{Proceedings of the IEEE}, vol.
  95:1, 2007, pp. 215 -- 233.

\bibitem{consensusTsitsiklis}
V.~D. Blondel, J.~M. Hendrickx, A.~Olshevsky, and J.~N. Tsitsiklis,
  ``Convergence in multiagent coordination, consensus, and flocking,'' in
  \emph{IEEE Conference on Decision and Control}, 2006, pp. 2996 -- 3000.

\bibitem{randomizedGossip}
S.~Boyd, A.~Ghosh, B.~Prabhakar, and D.~Shah, ``Randomized gossip algorithms,''
  \emph{IEEE Transactions on Information Theory}, vol.~52, pp. 2508--2530,
  2006.

\bibitem{BroadcastGossip}
T.~C. Aysal, M.~E. Yildiz, A.~D. Sarwate, and A.~Scaglione, ``Broadcast gossip
  algorithms for consensus,'' \emph{IEEE Transactions on Signal Processing},
  vol.~57, no.~7, pp. 2748 -- 2761, July 2009.

\bibitem{WeakErgodicityConditionsJadbabaie07}
A.~Tahbaz-Salehi and A.~Jadbabaie, ``Necessary and sufficient conditions for
  consensus over random independent and identically distributed switching
  graphs,'' in \emph{Proceedings of the 46th IEEE Conference on Decision and
  Control}, 2007.

\bibitem{AsymptoticConsensusVal}
V.~M. Preciado, A.~Tahbaz-Salehi, and A.~Jadbabaie, ``On asymptotic consensus
  value in directed random networks,'' in \emph{49th IEEE Conference on
  Decision and Control}, Atlanta, GA, USA, December 2010.

\bibitem{WeightedGossip}
F.~Benezit, V.~Blondel, P.~Thiran, J.~Tsitsiklis, and M.~Vetterli, ``Weighted
  gossip: Distributed averaging using non-doubly stochastic matrices,'' in
  \emph{IEEE International Symposium on Information Theory Proceedings (ISIT)},
  2010, pp. 1753 -- 1757.

\bibitem{trecateConsensus}
P.-A. Bliman and G.~Ferrari-Trecate, ``Average consensus problems in networks
  of agents with delayed communications,'' \emph{Automatica}, vol.~44, 2008.

\bibitem{richardDelayOverview}
J.-P. Richard, ``Time-delay systems: an overview of some recent advances and
  open problems,'' \emph{Automatica}, vol.~39, pp. 1667--1694, 2003.

\bibitem{saberDelays}
R.~Olfati-Saber and R.~M. Murray, ``Consensus problems in networks of agents
  with switching topology and time-delays,'' \emph{IEEE Transactions on
  Automatic Control}, vol.~49, no.~9, pp. 1520--1533, September 2004.

\bibitem{delayedConsensus}
A.~Seuret, D.~V. Dimarogonas, and K.~H. Johansson, ``Consensus under
  communication delays,'' in \emph{Proceedings of the 47th IEEE Conference on
  Decision and Control}, 2008.

\bibitem{bertsekasParallel}
D.~P. Bertsekas and J.~N. Tsitsiklis, \emph{Parallel and distributed
  computation: numerical methods}, 1st~ed.\hskip 1em plus 0.5em minus
  0.4em\relax Upper Saddle River, NJ, USA: Prentice-Hall, Inc., 1989.

\bibitem{delaysCao}
M.~Cao, S.~A. Morse, and B.~D.~O. Anderson, ``Reaching a consensus in a
  dynamically changing environment: Convergence rates, measurement delays, and
  asynchronous events,'' \emph{SIAM Journal on Control and Optimization},
  vol.~47, pp. 601--623, 2008.

\bibitem{nedicConsensusDelays}
A.~Nedic and A.~Ozdaglar, ``Convergence rate for consensus with delays,''
  \emph{Journal of Global Optimization}, vol.~47, no.~3, pp. 437--456, 2010.

\bibitem{VaidyaDelays2}
N.~H. Vaidya, C.~N. Hadjicostis, and A.~D. Dominguez-Garcia, ``Distributed
  algorithms for consensus and coordination in the presence of packet-dropping
  communication links - part ii: Coefficients of ergodicity analysis
  approach,'' UIUC, Tech. Rep., 2011.

\bibitem{SenetaBook}
E.~Seneta, \emph{Non-negative Matrices and Markov Chains}.\hskip 1em plus 0.5em
  minus 0.4em\relax Springer, 1973.

\bibitem{CortesDoublyStochDigraphs}
B.~Gharesifard and J.~Cortes, ``When does a digraph admit a doubly stochastic
  adjacency matrix?'' in \emph{Proceedings of the American Control Conference},
  Baltimore, Maryland, 2010, pp. 2440--2445.

\bibitem{TsitsiklisConvergenceSpeed09}
A.~Olshevsky and J.~N. Tsitsiklis, ``Convergence speed in distributed consensus
  and averaging,'' \emph{SIAM Journal on Control and Optimization}, vol. 48, No
  1, pp. 33--55, 2009.

\bibitem{TsianosACC2012}
K.~I. Tsianos and M.~G. Rabbat, ``Distributed dual averaging for convex
  optimization under communication delays,'' in \emph{American Control
  Conference (ACC)}, 2012.

\bibitem{fillBounds}
J.~A. Fill, ``Eigenvalue bounds on convergence to stationarity for non
  reversible markov chains, with an application to the exclusion process,''
  \emph{The Annals of Applied Probability}, vol.~1, no.~1, pp. 62--87, 1991.

\bibitem{StookDiaconis}
P.~Diaconis and D.~Stroock, ``Geometric bounds for eigenvalues of markov
  chains,'' \emph{The Annals of Applied Probability}, vol.~1, no.~1, pp.
  36--61, 1991.

\bibitem{contractingMatrices}
C.~W. Wu, ``On some properties of contracting matrices,'' \emph{Linear Algebra
  and its Applications}, vol. 428, pp. 2509--2523, 2008.

\bibitem{TsianosCDC2012}
K.~I. Tsianos, S.~Lawlor, and M.~G. Rabbat, ``Push-sum distributed dual
  averaging for convex optimization,'' in \emph{51st IEEE Conference on
  Decision and Control}, 2012.

\end{thebibliography}

\end{document}